\documentclass[10pt, a4paper]{article}

\usepackage{prooftree}
\usepackage{amssymb}
\usepackage{graphicx}
\usepackage{amsthm}
\usepackage[latin1]{inputenc}

\newcommand\ctx{\textsc{ctx}} 
\newcommand\lto{\mathbin{\backslash}}
\newcommand\lfrom{\mathbin{/}}
\newcommand\llto{\mathbin{\multimap}}
\newcommand{\resp}{\textit{resp.}~}
\newcommand\seq\vdash 
\def\N{\mbox{I\hspace{-.15em}N}}
\newcommand{\tuple}{\vec{x}}
\newcommand{\tuplebis}{\vec{y}}

\newtheorem{propri}{Property}
\newtheorem{defi}{Definition}
\newtheorem{prop}{Proposition}
\newtheorem{theo}{Theoreme}
\newtheorem{lem}{Lemma}

\title{Normalization and sub-formula property for Lambek with product and PCMLL -- Partially Commutative Multiplicative Linear Logic\label{mixtenorme}}
\author{Maxime Amblard, Christian Retor\'e\\
Loria (UMR 7503) Université de Lorraine, CNRS, INRIA Nancy Grand-Est\\
Labri -(UMR 5800) INRIA Bordeaux Sud-Ouest - Universit\'e de Bordeaux\\
amblard@loria.fr retore@labri.fr
}
\date{}
\begin{document}
\maketitle

\begin{verse}
\textbf{Abstract}

This paper establishes the normalisation of  natural deduction or lambda calculus formulation of 
Intuitionistic Non Commutative Logic --- which involves both commutative and non commutative connectives. This calculus first introduced by de Groote and as opposed to the classical version by Abrusci and Ruet admits a full entropy which allow order to be relaxed into any suborder. 
Our result also includes, as a special case, the normalisation of natural deduction the Lambek calculus with product, which is unsurprising but yet unproved. Regarding Intuitionistic Non Commutative Logic with full entropy does not have up to now a proof net syntax, and that for linguistic applications, sequent calculi which are only \emph{more or less} equivalent to natural deduction, are not convenient because they lack the standard Curry-Howard isomorphism. 
\end{verse}

\textbf{Keys words :}
Logic ; Intuitionistic Non Commutative Logic ; Lambek calculus ; normalisation

\newpage
\tableofcontents
\newpage

We first define partially commutative linear logic as a natural deduction system, Lambek calculus with product being an easily identified fragment, \cite{Ret05lcg}. 
Next we sketch its use for a logical account of minimalist grammars. 

Thereafter we prove normalisation for Lambek calculus with product which is a folklore result, assumed by the community but without any known proof. This enlightens the treatment of the product elimination rules to be intensively used for the complete calculus  in combination with order rules. 

Finally we give the proof of normalisation for the complete partially commutative calculus of de Groote in natural deduction. In addition to the difficulty of commutative conversion for product elimination rules, we are also faced with commutative conversions regarding the entropy rule acting on partial orders, \cite{AM07th}. 

We conclude with a discussion of normal form which we already know to be unique for the Lambek calculus with product and to be "nearly unique"  for the complete partially commutative calculus
(they are all obtained by permuting commutative-product elimination among sequences of commutative-product eliminations). But this would have made the paper too lengthy.

We begin by introducing and studying the properties of mixed logic, introduced in \cite{Gro96}, \cite{GR96} and \cite{BGR97}. 
This type of logic stems from the work of G. Gentzen, \cite{Gen34, Gen34b}, \cite{Gen36}, depending on the direction given by Lambek, \cite{Lam58}. For a detailed presentation of the evolution of these theories see \cite{GLT88}, \cite{Girard97un}.
Then, we present a typical use of this calculus for computational linguistic.
The two last sections present the normalization for Lambek with product and prove that the sub-formula property holds and the same for  PCMLL.


\section{Presentation}

Non-commutative logics arise naturally both in the mathematical perspective and in the modelling of some real world phenomena. 
Mathematically non commutativity is a natural both from the truth valued semantics viewpoint (phase semantics, based on monoids which can be non commutative) 
and from a syntactical one (sequent calculus with sequences rather than sets of formulae, proof nets which can have well bracketed axiom links). 
Non commutativity also appears from real world application such as concurrency theory, like concurrent execution of Petri net,  
and in our favourite application, computational linguistic, and this goes back  to the fifties and the apparition of the Lambek calculus. 
We first give a brief presentation of non commutative logics and then stress their interest for concurrency and computational linguistics. 

\subsection{Non commutative linear logics}

Linear logic \cite{Girard87} offered a logical view of the Lambek calculus \cite{Lam58} and non commutative calculi. 
During many years, the difficulty was to integrate commutative connective and non commutative connectives. A first solution, without a term calculus, was Pomset Logic, now studied with extended sequent calculi callled Calculus of Structures \cite{Guglielmi}.

Another kind of calculus was introduced as a sequent calculus by de Groote in \cite{Gro96}, which has to be intuitionistic to work neatly. It  
consists in a superposition of the Lambek calculus 
(non commutative) and of Intuitionnistic Linear Logic (commutative). For making a distinction bewteen the two connectives it is necessary that the context includes two different commas mimicking the conjunctions, one being commutative and the other being non commutative. Hence we deal with series-parallel partial orders over multisets of formulae as sequent right hand side. Let us  write $(...,...)$ for the parallel composition and $\langle...;...\rangle$ for the non commutative: hence 
$\langle\{a,b\};\{c,d\}\rangle$ stands for the finite partial order $a{<}c, b{<}c, a{<}d, b{<}d$. 
Of course we would like the two conjunctions to be related. Either the commutative product is stronger than the non commutative one, or the other way round. Surprisingly, the two options work as well, provided one direction is fixed once and for all. This relationship between the two products results from a structural rule modifying the order. 

Now a difference should be underlined between the Abrusci-Ruet classical calculus and the intuitionistic one of de Groote's and concern precisely the order rule. 
The Abrusci-Ruet definitely has an intuitionistic version by limiting it  to sequents with a single formula on the right and intuitionistic connectives, namely implications and conjunctions \cite{AR99apal}. 
But there is an important difference with de Groote's calculus: what can be the order rule stigmatising the relation between the two conjunctions? 

\begin{center}
\prooftree
	\Gamma\makebox{ ordered by $I$ } \vdash C 
	\justifies 
	\Gamma\makebox{ ordered by $J$ }\vdash C
\endprooftree
\end{center}

In de Groote's calculus $J$ can be any order such that $J \subset I $  (as set of ordered pairs of formulae in $\Gamma$) while in Ruet $J$ can only be obtained by turning some  non commutative commas into commutative ones. This not equivalent to allowing as result any suborder as $J$. Indeed  Bechet de Groote and the second author  in \cite{BGR97} showed that four rewriting rules are needed to obtained all possible series-parallel partial suborders from some series parallel partial order. Here is a typical derivation that can be performed in de Groote's calculus and not in Ruet's. 

\begin{center}
\prooftree
	(\{a,b\};\{c,d\})\vdash (a\otimes b) \odot (c \otimes d)
	\justifies 
	\{(\{a,b\};\{c),d\} \vdash (a\otimes b) \odot (c \otimes d)
\endprooftree
\end{center}

Abrusci-Ruet calculus admit a proof net syntax, which can be restricted to the intuitionistic case. 
Regarding the more flexible de Groote calculus, there neither exists proof nets, nor natural deduction:
it only exists a sequent calculus which has been prove to enjoy cut-elimination by semantical means in \cite{Gro96} and by a proof-theoretical method in \cite{Ret04petri}. 
This is what we propose in this paper, with normalisation. 
Firstly we thus obtain a calculus which is more convenient, because of the Curry-Howard isomorphism for computational linguistics, and secondly it is possibly a first step towards proof net syntax.

\subsection{Motivation for such calculi}

Non commutativity  in logic is rather natural in a resource consumption perspective. An hypothesis is viewed as a resource that can be use but then it is natural to think of how hypotheses are organised and accessible. As argued by Abrusci \cite{AbrusciJSL} and others, linearity is a mandatory condition for non commutativity. Observe that the first non commutative calculus, Lambek calculus which was invented long before linear logic, is a linear calculus, whose relation to oher logical system, in particular intuitionistic has only been understood after the invention of linear logic by Girard.

Concurrency, in which the order of the computations or of the resources matters, is of course a natural application. 
In the framework of proofs as programs, with normalisation as the computational process; it is rather the pomset logic and the calculus of structure which are of some use, because the order applies to  cuts that are the computations to be performed \cite{Ret97tlca,Guglielmi} But in the framework of proof search as computation, in the logic programming style of Miller, or in planning, the calculus studied in this paper with an order on hypotheses is of course important and process calculi can be encoded in non commutative calculi and this was the main motivation for Ruet's work.  The second author also provided a description of the parallel execution of a Petri net  in the calculus we are studying. It is a true concurrency approach,
where $a || b$ is not reduced to $a;b \oplus b;a$ (where $\oplus$ is the non deterministic choice). 
An execution according to a series parallel partial order corresponds to a proof in the partially commutative calculus that we study in this paper; in this order based approach of parallel computations any set of minimal transitions can be fired simultaneously. \cite{Ret04petri}

Our main motivation for such calculi is computational linguistics and grammar formalisms, 
in particular the description of mildly context formalisms. They are assumed to be large enough for natural language constructs, go beyond context-free languages, but admit polynomial parsing algorithms. 
We are especially fond of logical description of grammar classes as introduced by Lambek  because from a parse structure one is able to automatically compute the logical structure of the sentence. 
This especially true if the Lambek calculus or the partially commutative extensions that  we are considering are given in a natural deduction format. Indeed, the syntactic categories can be turned into semantic categories on two types, individuals ($e$) and truth values ($t$), in such a way that the proof in the Lambek calculus (the syntactic analysis) can be turned into a proof intuitionistic logic, that is a lambda term describing a logical formula in Church's style.

Lambek calculus is definitely too restrictive as a syntactic formalism, in particular it only describes context free languages, and many common syntactic constructs are difficult to model. This is the reason to use partially commutative calculi. In particular Lecomte and the second author managed to give a logical presentation \cite{LR01acl} of Stabler's minimalist grammars \cite{Sta97} in this the de Groote calculus, presented in natural deduction to obtain semantic representation of the parsed sentences. In parsing as deduction paradigm and for other applications as well it is quite important to have normalisation, unicity of the normal form: indeed the normal form is the structure of the analysed sentence, and normalisation ensures the coherence of the calculus. 
The algorithm of normalisation, easily extracted from the proof  is important as well:  
one define correct sentences as the ones such that some sequent can be proved, 
and both the parse structure and the semantic reading are obtained from the normal form.

\section{Partially Commutative Linear Logic}

\subsection{Order and formulae}
The sequent calculus  for Partially Commutative Linear Logic (PCIMLL) was introduced by de Groote in  \cite{Gro96}. It is a super imposition of between commutative intuitionistic multiplicative linear logic and the Lambek calculus with product, that is non commutative intuitionistic multiplicative linear logic. Formulae are defined from a set of propositional variables \textsc{p}, by the commutative conjunction ($\otimes$) , the non commutative conjunction ($\odot$), the commutative implication ($\llto$), the two non commutative implications ($\lfrom $ and $\lto$): 
 $$\textsc{l}::= \textsc{p} \;|\; \textsc{l}\odot \textsc{l} \;|\; \textsc{l}\otimes \textsc{l} \; |\; \textsc{l}\lfrom \textsc{l} \; | \; \textsc{l}\lto \textsc{l} \;|\; \textsc{l} \llto \textsc{l}$$

 Left hand side are partially ordered multiset of formulae whose underlying order is series-parallel (sp), that is can be defined by two operation: disjoint union denoted by $(\_,\_)$ 
and conjoint union, denoted by $\langle \_;\_\rangle$ (the domain is the disjoint union of the two domain, and every formulae in the first component is before any formulae in the second 
component. They respect the following syntax:

  $$\ctx ::= \textsc{l} | \langle \ctx ; \ctx \rangle \;|\; (\ctx,\ctx )$$

For example, the context $\langle \langle B;(A\llto (B\lto (D\lfrom C), A) \rangle ; C\rangle$ denote the sp order, $Succ(B)=(A,A\llto (B\lto (D\lfrom C) )$,
$Succ(A)=Succ(A\llto (B\lto (D\lfrom C))= C$ where $Succ(X)$ is the immediate successor of $X$ and this function of the domain to parts of this domain determine a full finite order.

The term denoting an sp order is unique up to the commutativity of $(\_,\_)$  and to the associativity of both  $(\_,\_)$ and $\langle \_;\_\rangle$. The term notation is only a short hand, a convenient notation for sp order. That is even if the sp term are different, the left hand side of two sequents 
are considered as equal whenever they are equal as partially ordered multiset. 
 
Uppercase Greek letters are used for contexts.
An expression $\Gamma[]$ represents a context in which we distinguish a specific element $[*]$, where for an expression $\Gamma[\Delta]$ the element $[*]$ is replaced by the context $\Delta$. 
More details can be found in \cite{BGR97} and \cite{Ret04petri}.

Figure \ref{reglesmixtes} shows all rules of PCMLL. It uses the classical rules of commutative multiplicative intuitionistic linear logic and of non-commutative multiplicative intuitionistic linear logic. Both bring introductions and eliminations for its implicative(s) connectives and its product connective. Moreover, we use the axiom rule and an entropy rule ($\sqsubset$) which correspond to order's inclusions (weakness of order).

\begin{figure}
\begin{small}

\prooftree
	\Gamma \vdash A
	\quad \Delta \vdash  A \lto C
	\justifies
	<\Gamma; \Delta> \vdash C 
	\using
             [\lto_{e}]
         \thickness=0.07em
\endprooftree
\hfill
\prooftree
	\Delta \vdash  A \lfrom  C
	\quad \Gamma \vdash A
	\justifies
	< \Delta ; \Gamma> \vdash C 
	\using
             [\lfrom _{e}]
         \thickness=0.07em
\endprooftree
\hfill
\prooftree
	 \Gamma \vdash A
	\quad \Delta \vdash  A \llto C
	\justifies
	(\Gamma, \Delta) \vdash C 
	\using
             [\llto_{e}]
         \thickness=0.07em
\endprooftree

\bigskip 

\prooftree
	<A; \Gamma> \vdash C
	\justifies
	\Gamma \vdash A \lto C 
	\using
             [\lto_{i}]
         \thickness=0.07em
\endprooftree
\hfill
\prooftree
	<\Gamma; A> \vdash  C
	\justifies
	\Gamma \vdash C\lfrom A 
	\using
             [\lfrom _{i}]
         \thickness=0.07em
\endprooftree
\hfill
\prooftree
	(A, \Gamma ) \vdash  C
	\justifies
	\Gamma \vdash A \llto C 
	\using
             [\llto_{i}]
         \thickness=0.07em
\endprooftree

\bigskip 

\prooftree
	\Delta \vdash  A \odot B
	\quad \Gamma, <A; B>, \Gamma ' \vdash C
	\justifies
	\Gamma, \Delta , \Gamma '  \vdash C 
	\using
             [\odot_{e}]
         \thickness=0.07em
\endprooftree
\hfill 
\prooftree
	\Delta \vdash  A \otimes B
	\quad \Gamma, (A, B),  \Gamma ' \vdash C
	\justifies
	\Gamma, \Delta , \Gamma '  \vdash C 
	\using
             [\otimes_{e}]
         \thickness=0.07em
\endprooftree

\bigskip 

\prooftree
	\Delta \vdash A
	\quad \Gamma \vdash B
	\justifies
	<\Delta; \Gamma> \vdash A \odot B
	\using
             [\odot_{i}]
         \thickness=0.07em
\endprooftree
\hfill
\prooftree
	\Delta \vdash A
	\quad \Gamma \vdash B
	\justifies
	(\Delta, \Gamma) \vdash A \otimes B
	\using
             [\otimes_{i}]
         \thickness=0.07em
\endprooftree

\prooftree
	\justifies
	A \vdash A
	\using
             [axiom]
         \thickness=0.07em
\endprooftree
\hfill 
\prooftree
	\Gamma \vdash C
	\justifies
	\Gamma' \vdash C 
	\using
             [\mbox{entropy --- whenever }\Gamma'\sqsubset\Gamma]
         \thickness=0.07em
\endprooftree
\end{small} 
\caption{Rules of PCMLL\label{reglesmixtes}.}
\end{figure}

This calculus deserve some explanation and comments, especially the later rule: the entropy:

$\Gamma'\sqsubset\Gamma$ whenever these contexts that are sp partiallly ordered multisets of formulae 
have the same multiset domain $|\Gamma|=|\Gamma'|$, and whenever considering each occurrence of a formula as distinct if $A<B$ in $\Gamma'$  then $A<B$ in $\Gamma$ as well. 
The inclusion $\sqsubset$  of series parallel partial orders can be viewed as a rewriting relation (modulo commutativity and associativity)  on the sp term denoting them as shown in \cite{BGR97} -- see also \cite{Ret04petri} where the order rule is used the other way round, but as said in the introduction, it does not change normalisation. 

In the $\otimes$ and $\odot_e$ rules, A and B must be equivalent:$$\forall X \neq A,B
\left\{ \begin{array}{l}
X<A \Leftrightarrow X<B\\ 
X>A \Leftrightarrow X>B
\end{array} \right. $$ 

 In the $\otimes_e$ case A and B are equivalent and uncomparable while in the $\odot_e$ case, they are equivalent and $A<B$. In the conclusion of rules, they are replaced by the context which had produced $A\otimes B$(i.e. in figure \ref{reglesmixtes} $\Delta$).

Our formulation in a lambda calculus style of the elimination of the multiplicative linear logic conjunction is due Abramsky in \cite{Abramsky93cill} --- on the term side it corresponds to $let\; x=(u,v) in \; t(u,v)$ construct. 

Although we do have normalisation and sub-formula property (next sections) we do not have complicated rules of the kind  introduced in \cite{Negri02nd} for MLL. We assume her rules are motivated by other properties as well, and work for the complete  linear calculus with additive and exponentials.

\subsection{General definitions}

Lets  $\delta$ be a proof,  $S_j$ denotes an instance of a sequent in $\delta$, $|S_j|$ denotes for the corresponding sequent, and $|S_j|^r$ the formula in the right hand side of this sequent.

In a proof $\delta$, $B(S_0)$ is the \textbf{principal branch} outcome of an occurrence $S_0$ os a sequent, $|B(S_0)|$ is the smallest path which contains $S_0$ and closed by the following: 
\begin{enumerate} 
\item 
If $S\in B(S_0)$ is obtained by an unary rule $R$ os an occurence $S'$ of a sequent, then $S'\in B(S_0)$.
\item 
If $S\in B(S_0)$ is obtained by a product elimination $\odot_e$ (\resp $\otimes_e$), then 
the premise which brings the marker $S'$, with $|S'|=\Gamma[\langle A,B\rangle]\seq C$ (\resp  $|S'|=\Gamma[(A,B)]\seq C$) is also in $B(S_0)$.

\item 
If $S\in B(S_0)$ is obtained by an implicative elimination rule $\lto_e$ (\resp $\lfrom_e$, $\llto_e$), then the premisewhich brings the marher $S'$, with $|S'|= \Delta \vdash  A \lto C$ 
(\resp $|S'|= \Delta \vdash  A \lfrom C$, $|S'|= \Delta \vdash  A \llto C$) is also in $B(S_0)$.

\end{enumerate} 

For every path in a principal branch $B(S)$ from $S$ to $S_i$ such that $|S|^r = |S_i|^r$, if $|S|$ is an elimination rule and $|S_i|$ an introduction one, they are over the same formula and they are said \textbf{conjoined}

\section{A short example using PCMLL in Computational Linguistics}

Before we prove the normalisation and subformula property of PCMLL, let us illustrate briefly our use of this calculus for computational linguistics
--- for more details see \cite{LR01acl}. 
As said above, Lambek calculus is too restricted to describe natural language syntax, hence we try to logically describe richer formalisms, 
the logic being an easy way to extract semantic readings that are higher order logical formula describing the meaning of the sentence.
We focused on Stabler's minimalist grammars because they have good computational properties (polynomial parsing) 
are rather close to Lambek grammars and formalise Chomsky's recent minimalist program (hence they inherit from a good coverage of lot of syntactic constructs), \cite{Sta97, AM07th, ALR04}.

As Lambek grammars our categorial minimalist grammars are lexicalised: a lexicon maps every word into a PCMLL formula which describes its interaction with other words. 
We do not use all the rules of PCMLL but only some of them, grouped in derived rules. 
There are strings of words and variables labelling every formula of every sequent in the proof. 
We first derive natural deduction trees from axioms $x:A\vdash x:A$ and proper axioms $\vdash w:T$ 
when $T$ is the type of the word $w$ yielding a proof of $sentence:C$.

\begin{center}
MERGE 
\begin{small}
\prooftree
	\prooftree
		\tuple:\Delta \vdash w:A
		\quad \tuple':\Gamma \vdash w':A \lto C 
		\justifies
		\tuple:\Delta ;\tuple':\Gamma \vdash w w': C 
		\using
	             [\lto_e]
	         	\thickness=0.07em
	\endprooftree
	\justifies
	\tuple:\Delta ,\tuple'\Gamma \vdash w w':  C 
	\using
             [entropy]
         \thickness=0.07em
\endprooftree
\end{small}
\hfill or \hfill 
\begin{small}
\prooftree
	\prooftree
		\quad \tuple':\Gamma \vdash  w':C \lfrom  A
		\quad \tuple \tuple:\Delta \vdash w:A
		\justifies
		\tuple':\Gamma ;\tuple:\Delta \vdash w' w:C
		\using
	             [\lfrom _e]
	         	\thickness=0.07em
	\endprooftree
	\justifies
\tuple':\Gamma ;\tuple:\Delta \vdash w' w:C
	\using
             [entropy]
         \thickness=0.07em
\endprooftree
\end{small}
\end{center}

The \emph{merge} rule is completely similar to the residuation law of AB grammars of Lambek grammars. 

\begin{center}
MOVE 
\begin{small}
\prooftree
	\tuplebis:\Gamma \vdash s: A \otimes B 
	\quad \tuple:\Delta, x: A, y: B, \tuple':\Delta ' \vdash t: C
	\justifies
	\tuple:\Delta ,\tuplebis:\Gamma, \tuple':\Delta ' \vdash t[s\lfrom x,\epsilon\lfrom y] : C 
	\using
             [\otimes_e]
         \thickness=0.07em
\endprooftree
\end{small}
\end{center}

The \emph{move} rule is typical from Chomskyan linguistics: our construct mimics the 
movement of the  constituent $\lfrom$ string $s$ from the $y$ place to the $x$ place.

Here is an example with a tiny lexicon, especially simple because Italian allows null subjects, 
exemplifying the movement of an interrogative noun phrase to a leftmost position: 

\begin{tabular}{p{1cm}p{3cm}p{1cm}p{2cm}p{1cm}p{2cm}}
que& $wh \otimes (k \otimes d)\lfrom n$
&cosa& $n$
&$\epsilon$ & $k \otimes d$\\
fai  &$k\lto d \lto v \lfrom  d$
&
infl &$k \lto t \lfrom  v$
&comp& $wh \lto c \lfrom  t$
\end{tabular}

\begin{figure}[ht]
\begin{tiny}
\hspace{-19ex}
\prooftree
	\prooftree
	\begin{array}{l}che\\ \vdash wh \otimes (k \otimes d)\lfrom n\end{array}
	\begin{array}{l}cosa\\\quad \vdash n\end{array}
	\justifies
	\vdash wh \otimes (k \otimes d)
	\using
	  [merge]
	\thickness=0.07em
	\endprooftree
	\prooftree
		wh \vdash wh
		\hspace{-3ex}
		\prooftree
			\begin{array}{l}comp\\ \vdash wh \lto c \lfrom  t\end{array}
			\hspace{-3ex}
			\prooftree
				 \begin{array}{l}\epsilon\\\vdash k \otimes d\end{array}
				 \hspace{-3ex}
				\prooftree
					k \vdash k
					\hspace{-3ex}
					\prooftree
						\begin{array}{l}infl\\\vdash k \lto t \lfrom  v\end{array}
						\prooftree
							d \vdash d
							\hspace{-3ex}
							\prooftree
								k \otimes d \vdash k \otimes d
								\hspace{-3ex}
								\prooftree
									k \vdash k
									\hspace{-3ex}
									\prooftree
										\begin{array}{l}fai\\ \vdash k\lto d \lto v \lfrom  d\end{array}
										\begin{array}{l}$\,$\\\quad d \vdash d\end{array}
										\justifies
										d \vdash k\lto d \lto v
										\using
										  [merge]
										\thickness=0.07em
									\endprooftree
									\justifies
									k, d \vdash d \lto v
									\using
									  [merge]
									\thickness=0.07em
								\endprooftree
								\justifies
								k \otimes d \vdash d \lto v
								\using
								  [move]
								\thickness=0.07em
							\endprooftree
							\justifies
							d, k \otimes d \vdash v
							\using
							  [merge]
							\thickness=0.07em
						\endprooftree
						\justifies
						d, k \otimes d \vdash k \lto t 
						\using
						  [merge]
						\thickness=0.07em
					\endprooftree
					\justifies
					k, d, k \otimes d \vdash t 
					\using
					  [merge]
					\thickness=0.07em
				\endprooftree
				\justifies
				k \otimes d \vdash t 
				\using
				  [move]
				\thickness=0.07em
			\endprooftree
			\justifies
			k \otimes d \vdash wh \lto c 
			\using
			  [merge]
			\thickness=0.07em
		\endprooftree
		\justifies
		wh, k \otimes d \vdash c 
		\using
		  [merge]
		\thickness=0.07em
	\endprooftree
	\justifies
	\vdash c 
	\using
	  [move]
	\thickness=0.07em
\endprooftree
\end{tiny}
\caption{"che cosa $\epsilon$ fai ?" analyze}
\end{figure}

Our interest for such an analysis is that we are able  from such a structure to automatically compute the semantic reading with correct quantifier scope, 
which is, for our example: $\exists? (\lambda x\ (\land (cosa(x)) (far(tu,x)))) \equiv \exists? x (cosa(x) \land far(tu,x)$ 
Of course our representation of minimalist grammars and its semantic-interpretation outcome is extremely sketchy, 
and details can be found in \cite{LR01acl,Amblard06fg}. Indeed the objective of this paper is the normalisation of this calculus.

\section{Normalisation du Lambek avec produit}

\subsection{Properties of L$_{\odot}$}
Lambek calculus with product (L$_\odot$) is the restriction of PCMLL to the connectives:  $\lto$, $\lfrom $ and $\odot$.
Furthermore, it uses only the order $\langle ... ; ... \rangle$, \textit{i.e.} it treats  sequences of formulae (total order). That wy, the order is not mark. The entropy rule is not use here. 

\begin{propri}\label{proofriselambek}

Let $R$ a product elimination $\odot_e$ of $\Gamma[\Delta]\seq C$ between a proof $\delta_0$ with $\Delta\seq A\odot B$ as  conclusion and a proof with $\Gamma[A , B ]\seq C$ as conclusion obtained by a rule $R'$ between a proof $\delta_1$ of a sequent $\Theta[ A , B]\seq X$ (and if $R'$ is a binary rule, a second proof $\delta_2$ with $\Psi\seq U$ as conclusion).
This description of the structure of theses proofs is presented in figure \ref{montprod1}.

Thus, we can obtain a proof for the same sequent $\Gamma[\Delta]\seq C$ by first apply the rule $R$ between the proof $\delta_0$ with $\Delta\seq A\odot B$ as conclusion and the proof $\delta_1$ with $\Theta[A,B]\seq X$ as conclusion, which produces the sequent $\Theta[\Delta]\seq X$.
By applying the rule $R'$ to this new proof and the proof $\delta_2$, we obtain the same sequent $\Gamma[\Delta]\seq C$.
\end{propri}

\begin{figure}[htbp]
\begin{center}

\prooftree
	\prooftree
		\leadsto
		\Delta \seq A\odot B
		\using
		  \delta_0
	\endprooftree
	\quad 
	\prooftree
		\prooftree
			\leadsto
			\Psi \seq U
			\using
			  \delta_2			
		\endprooftree
		\prooftree
			\leadsto
			\Theta[ A,B] \seq X
			\using
			  \delta_1			
		\endprooftree
		\justifies
		\Gamma[A,B] \seq C
		\using R'
	\endprooftree
	\justifies
	\Gamma[\Delta] \seq C
	\using R
\endprooftree
$\quad \Rightarrow \quad$
\prooftree
	\prooftree
		\leadsto
		\Psi \seq U
		\using
		  \delta_2			
	\endprooftree
	\prooftree
		\prooftree
			\leadsto
			\Delta \seq A\odot B
			\using
			  \delta_0
		\endprooftree
		\quad 
		\prooftree
			\leadsto
			\Theta[A,B] \seq X
			\using
			  \delta_1			
		\endprooftree
		\justifies
		\Theta[\Delta] \seq X
		\using R
	\endprooftree
	\justifies
	\Gamma[\Delta] \seq C
	\using R'
\endprooftree
\end{center}
\caption{Proof's structure allowing the rise of the product in L$_{\odot}$. \label{montprod1}}
\end{figure}

\begin{proof} 
The proof of this proposition is a case study based on the type of the rule above the product elimination. 
The following presents the different cases:

\begin{itemize}
\item[$\circ$] Rising over $\lto_e$:

\begin{itemize}
\item[$\bullet$] hypotheses in the left premise of $\lto_e$:

\bigskip

\hspace{-5ex}
\begin{tabular}{lcc}
\prooftree 
	\Gamma \vdash  A \odot B
	\; \prooftree 
		A, B \vdash D
		\quad \Delta \vdash D \lto C
		\justifies
		A,B, \Delta \vdash C 
		\using
	             [\lto_{e}]
	         \thickness=0.07em
	\endprooftree 
	\justifies
	\Gamma, \Delta \vdash C 
	\using
             [\odot_{e}]
         \thickness=0.07em
\endprooftree 
&
$\Rightarrow$
&
\prooftree 
	\prooftree 
		\Gamma \vdash  A \odot B
		\quad A, B \vdash D
		\justifies
		\Gamma \vdash D
		\using
	             [\odot_{e}]
	         \thickness=0.07em
	\endprooftree 
	\Delta \vdash D \lto C
	\justifies
	\Gamma, \Delta \vdash C 
	\using
             [\lto_{e}]
         \thickness=0.07em
\endprooftree \\
\end{tabular}

\bigskip

\item[$\bullet$] hypotheses in the right premise of $\lto_e$:

\bigskip

\hspace{-5ex}
\begin{tabular}{ccc}
\prooftree 
	\Gamma \vdash  A \odot B
	\; \prooftree 
		\Delta \vdash D
		\quad A, B \vdash D \lto C
		\justifies
		\Delta, A,B \vdash C 
		\using
	             [\lto_{e}]
	         \thickness=0.07em
	\endprooftree 
	\justifies
	\Delta , \Gamma \vdash C 
	\using
             [\odot_{e}]
         \thickness=0.07em
\endprooftree 
&
$\Rightarrow$
&
\prooftree 
	\Delta \vdash D 
	\quad \prooftree 
		\Gamma \vdash  A \odot B
		\quad A, B \vdash D \lto C
		\justifies
		\Gamma \vdash D \lto C
		\using
	             [\odot_{e}]
	         \thickness=0.07em
	\endprooftree
	\justifies
	\Delta , \Gamma \vdash C 
	\using
             [\lto_{e}]
         \thickness=0.07em
\endprooftree \\
\end{tabular}
\end{itemize}
\bigskip

\item[$\circ$] Rising over $\lto_i$:

\bigskip

\hspace{-3ex}
\begin{tabular}{lcc}
\prooftree 
	\Gamma \vdash  A \odot B
	\quad \prooftree 
		D, \Delta, A,B, \Delta ' \vdash C
		\justifies
		\Delta, A,B, \Delta ' \vdash D \lto C 
		\using
	             [\lto_{i}]
	         \thickness=0.07em
	\endprooftree 
	\justifies
	\Delta, \Gamma, \Delta ' \vdash D \lto C 
	\using
             [\odot_{e}]
         \thickness=0.07em
\endprooftree 
&
$\Rightarrow$
&
\prooftree 
	\prooftree 
		\Gamma \vdash  A \odot B
		\quad D, \Delta, A,B, \Delta '\vdash C
		\justifies
		D, \Delta, \Gamma, \Delta ' \vdash C
		\using
	             [\odot_{e}]
	         \thickness=0.07em
	\endprooftree 
	\justifies
	\Delta, \Gamma, \Delta ' \vdash D \lto C
	\using
             [\lto_{i}]
         \thickness=0.07em
\endprooftree 
\end{tabular}
\bigskip

\item[$\circ$] Rising over $\odot_e$:
\begin{itemize}

\item[$\bullet$] hypotheses in the left premise of $\lfrom _e$:
\smallskip

\hspace{-5ex}
\bigskip
\begin{tabular}{lcc}
\prooftree 
	\Gamma \vdash  A \odot B
	\prooftree 
		\Delta \vdash C \lfrom  D
		\quad A, B \vdash D
		\justifies
		\Delta, A,B \vdash C 
		\using
	             [\lfrom _{e}]
	         \thickness=0.07em
	\endprooftree 
	\justifies
	\Delta, \Gamma \vdash C 
	\using
             [\odot_{e}]
         \thickness=0.07em
\endprooftree 
&
$\Rightarrow$
&
\prooftree 
	\Delta \vdash C \lfrom  D
	\quad \prooftree 
		\Gamma \vdash  A \odot B
		\quad A, B \vdash D
		\justifies
		\Gamma \vdash D
		\using
	             [\odot_{e}]
	         \thickness=0.07em
	\endprooftree 
	\justifies
	\Delta, \Gamma \vdash C 
	\using
             [\lfrom _{e}]
         \thickness=0.07em
\endprooftree 
\end{tabular}
\bigskip

\item[$\bullet$] hypotheses in the right premise of $\lfrom _e$:

\bigskip

\hspace{-5ex}
\begin{tabular}{lcc}
\prooftree 
	\Gamma \vdash  A \odot B
	\prooftree 
		A, B \vdash C \lfrom  D
		\quad \Delta \vdash D
		\justifies
		A,B, \Delta \vdash C 
		\using
	             [\lfrom _{e}]
	         \thickness=0.07em
	\endprooftree 
	\justifies
	\Gamma, \Delta \vdash C 
	\using
             [\odot_{e}]
         \thickness=0.07em
\endprooftree 
&
$\Rightarrow$
&
\prooftree 
	\prooftree 
		\Gamma \vdash  A \odot B
		\quad A, B \vdash C \lfrom  D
		\justifies
		\Gamma \vdash C \lfrom  D
		\using
	             [\odot_{e}]
	         \thickness=0.07em
	\endprooftree 
	\Delta \vdash D
	\justifies
	\Gamma, \Delta \vdash C 
	\using
             [\lfrom _{e}]
         \thickness=0.07em
\endprooftree 
\end{tabular}
\end{itemize}
\smallskip

\item[$\circ$] Rising over $\lfrom _i$:

\bigskip

\hspace{-5ex}
\begin{tabular}{lcc}
\prooftree 
	\Gamma \vdash  A \odot B
	\quad \prooftree 
		\Delta, A,B, \Delta ', D \vdash C
		\justifies
		\Delta, A,B, \Delta ', \vdash C \lfrom  D 
		\using
	             [\lfrom _{i}]
	         \thickness=0.07em
	\endprooftree 
	\justifies
	\Delta, \Gamma, \Delta ' \vdash C \lfrom  D
	\using
             [\odot_{e}]
         \thickness=0.07em
\endprooftree 
&
$\Rightarrow$
&
\prooftree 
	 \prooftree 
		\Gamma \vdash  A \odot B
		\quad \Delta, A,B, \Delta ', D \vdash C
		\justifies
		\Delta, \Gamma, \Delta ', D \vdash C
		\using
	             [\odot_{e}]
	         \thickness=0.07em
	\endprooftree 
	\justifies
	\Delta, \Gamma, \Delta ' \vdash C \lfrom  D
	\using
             [\lfrom _{i}]
         \thickness=0.07em
\endprooftree 
\end{tabular}

\bigskip
\item[$\circ$] Rising over $\odot_e$:
\begin{itemize}
\item[$\bullet$] hypotheses in the left premise of the first $\odot_e$:

\bigskip
\prooftree 
	\Gamma \vdash A \odot B
	\prooftree 
		\Delta, A,B, \Delta ' \vdash C \odot D
		\quad \Phi, C,D, \Phi ' \vdash E
		\justifies
		\Phi, \Delta, A,B, \Delta ', \Phi ' \vdash E
		\using
	             [\odot_{e}]
	         \thickness=0.07em
	\endprooftree 
	\justifies
	\Phi, \Delta, \Gamma, \Delta ', \Phi ' \vdash E
	\using
             [\odot_{e}]
         \thickness=0.07em
\endprooftree \\

\bigskip

\begin{flushright}
$\Rightarrow \quad$
\prooftree 
	\prooftree 
		\Gamma \vdash A \odot B
		\quad \Delta, A,B, \Delta ' \vdash C \odot D
		\justifies
		\Delta, \Gamma, \Delta ' \vdash C \odot D
		\using
	             [\odot_{e}]
	         \thickness=0.07em
	\endprooftree 
	\quad \Phi, C,D, \Phi ' \vdash E
	\justifies
	\Phi, \Delta, \Gamma, \Delta, \Phi ' \vdash E
	\using
             [\odot_{e}]
         \thickness=0.07em
\endprooftree 
\end{flushright}
\bigskip

\item[$\bullet$] hypotheses in the right premise of the first $\odot_e$:

\bigskip
\prooftree 
	\Gamma \vdash A \odot B
	\prooftree 
		\Delta \vdash C \odot D
		\quad \Phi, A,B,C,D, \Phi ' \vdash E
		\justifies
		\Phi, A,B,\Delta, \Phi ' \vdash E
		\using
	             [\odot_{e}]
	         \thickness=0.07em
	\endprooftree 
	\justifies
	\Phi, \Gamma,\Delta, \Phi ' \vdash E
	\using
             [\odot_{e}]
         \thickness=0.07em
\endprooftree \\

\begin{flushright}
$\Rightarrow \quad$
\prooftree 
	\Delta \vdash C \odot D
	\prooftree 
		\Gamma \vdash A \odot B
		\quad \Phi, A,B,C,D, \Phi ' \vdash E
		\justifies
		\Phi, \Gamma,C,D, \Phi ' \vdash E
		\using
	             [\odot_{e}]
	         \thickness=0.07em
	\endprooftree 
	\justifies
	\Phi, \Gamma,\Delta, \Phi ' \vdash E
	\using
             [\odot_{e}]
         \thickness=0.07em
\endprooftree 
\end{flushright}
\end{itemize}
\bigskip

\item[$\circ$] Rising over $\odot_i$:
\begin{itemize}
\item[$\bullet$] hypotheses in the left premise of the first $\odot_i$:

\bigskip
\hspace{-9ex}
\prooftree 
	\Gamma \vdash A \odot B
	\prooftree 
		\Delta, A,B, \Delta ' \vdash C
		\quad \Phi \vdash D
		\justifies
		\Delta, A,B, \Delta ', \Phi \vdash C \odot D
		\using
	             [\odot_{i}]
	         \thickness=0.07em
	\endprooftree 
	\justifies
	\Delta, \Gamma, \Delta ', \Phi \vdash C \odot D
	\using
             [\odot_{e}]
         \thickness=0.07em
\endprooftree 
$\;\, \Rightarrow$
\prooftree 
	\prooftree 
		\Gamma \vdash A \odot B
		\quad \Delta, A,B, \Delta ' \vdash C
		\justifies
		\Delta, \Gamma, \Delta ' \vdash C
		\using
	             [\odot_{e}]
	         \thickness=0.07em
	\endprooftree 
	\Phi \vdash D
	\justifies
	\Delta, \Gamma , \Delta ', \Phi \vdash C \odot D
	\using
            [\odot_{i}]
         \thickness=0.07em
\endprooftree 

\bigskip
\item[$\bullet$] hypotheses in the right premise of the first $\odot_i$:

\bigskip
\hspace{-10ex}
\begin{tabular}{lcc}
\prooftree 
	\Gamma \vdash A \odot B
	\prooftree 
		\Delta \vdash C
		\quad \Phi, A,B, \Phi ' \vdash D
		\justifies
		\Delta, \Phi, A,B, \Phi ' \vdash C \odot D
		\using
	             [\odot_{i}]
	         \thickness=0.07em
	\endprooftree 
	\justifies
	\Delta, \Phi, \Gamma, \Phi ' \vdash C \odot D
	\using
             [\odot_{e}]
         \thickness=0.07em
\endprooftree 
&
$\Rightarrow$
&
\prooftree 
	\Delta \vdash C
	\prooftree 
		\Gamma \vdash A \odot B
		\quad  \Phi, A,B, \Phi ' \vdash D
		\justifies
		\Phi, \Gamma, \Phi ' \vdash D
		\using
	             [\odot_{e}]
	         \thickness=0.07em
	\endprooftree 
	\justifies
	\Delta, \Phi, \Gamma, \Phi ' \vdash C \odot D
	\using
            [\odot_{i}]
         \thickness=0.07em
\endprooftree
\bigskip
 
\end{tabular}
\end{itemize}
\end{itemize}

All possible cases of combinations of rules have been examined. 
The product elimination has the ability to rise above any rule if the hypothesis used are in the same premise.
\end{proof}

\begin{defi}
A $\odot_e$ rule over the sequent $A\odot B$ is said \textbf{as high as possible} if in the premise of the rule above, the two hypothesis $A$ and $B$ are not in the same formula.

We say that a rule is \textbf{at the bottom} of the proof if it is the rule that provides the sequent conclusion of the proof.
\end{defi}

Let us call a \textbf{redex} in a proof the immediate succession of an introduction rule and its conjoined elimination.

In this calculus, there are four \textit{redexes}: $\lfrom$ one, $\lto$ one and $\odot$ two, depending in which premise the $\odot_i$ takes place.
We give the proof schemes
\begin{enumerate}
\item[$\circ$] Redex$_{\lfrom }$: introduction $\lfrom _i$ and immediate elimination of $\lfrom _e$.
\begin{center}
\prooftree 
	\prooftree 
		\prooftree
			D
			\leadsto 
			C
		\endprooftree
		\justifies
		C \lfrom D 
		\using
	         	    [\lfrom _{i}]
         		\thickness=0.07em
	\endprooftree 
	\quad \prooftree \leadsto  D  \using \delta_1 \proofdotnumber=5 \endprooftree 
	\justifies
	C
	\using
             [\lfrom _{e}]
         \thickness=0.07em
\endprooftree 
$\quad \Rightarrow\quad$
\prooftree \prooftree   \leadsto D \using \delta_1 \proofdotnumber=5 \endprooftree \leadsto  C\endprooftree
\end{center}
\item[$\circ$] Redex$_{\lto}$: introduction $\lto_i$ and immediate elimination of $\lto_e$.
\begin{center}
\prooftree 
	\prooftree \leadsto  D \using \delta_1 \proofdotnumber=5 \endprooftree 
	\quad\prooftree 
		\prooftree D\leadsto  C\endprooftree 
		\justifies
		D\lto C 
		\using
	         	    [\lto_{i}]
         		\thickness=0.07em
	\endprooftree 
	\justifies
	C
	\using
             [\lto_{e}]
         \thickness=0.07em
\endprooftree 
$\quad \Rightarrow\quad$
\prooftree \prooftree \leadsto D \using \delta_1 \proofdotnumber=5 \endprooftree \leadsto  C\endprooftree
\end{center}

\item[$\circ$] Redex$_{\odot}$: introduction $\odot_i$ and immediate elimination of $\odot_e$ on the left. gauche.
\begin{center}
\prooftree 
	\prooftree 
		\prooftree \leadsto  A\using \delta_1 \proofdotnumber=5 \endprooftree 
		\prooftree \leadsto  B\using \delta_2 \proofdotnumber=5 \endprooftree 
		\justifies
		A\odot B
		\using
	         	    [\odot_{i}]
         		\thickness=0.07em
	\endprooftree 
	\quad \prooftree A \quad B \leadsto  D\endprooftree 
	\justifies
	D
	\using
             [\odot_{e}]
         \thickness=0.07em
\endprooftree 
$\quad \Rightarrow\quad$
 \prooftree  \prooftree \leadsto  A\using \delta_1 \proofdotnumber=5\endprooftree  \quad \prooftree  \leadsto  B \using \delta_1 \proofdotnumber=5 \endprooftree  \leadsto  D\endprooftree
\end{center}
\item[$\circ$] Redex$_{\odot}$: introduction $\odot_i$ and immediate elimination of $\odot_e$ on the right.
\begin{center}
\prooftree 
	 \prooftree \leadsto  A\odot B \using \delta_1 \proofdotnumber=5 \endprooftree 
	\quad \prooftree 
		A
		\quad B
		\justifies
		A\odot B
		\using
	         	    [\odot_{i}]
         		\thickness=0.07em
	\endprooftree 
	\justifies
	A\odot B
	\using
             [\odot_{e}]
         \thickness=0.07em
\endprooftree 
$\quad \Rightarrow\quad$
\prooftree \leadsto  A\odot B \using \delta_1 \proofdotnumber=5\endprooftree
\end{center}
\end{enumerate}

From the notion of redex, we define the \textbf{$\mathbf{k}$-extended-redex} on a proof.

Every path of a principal branch $B(S_0)$ of lenght $k$ from $S_0$ to $S_n$ with $|S_0|^r=|S_n|^r$, such that $|S_0|$ marks an elimination rule $R_e$ and $S_n$ is the conclusion of an introduction rule $R_i$, is called a $\mathbf{k}$\textbf{-extended-redex}. 
Note that $0$\emph{-extended-redexes} are redexes of PCMLL, already presented.

\begin{prop} 
A $k$-extended-redex only containts $\odot_e$ rules or a $k'$-extended-redex, with $k' < k$.
\end{prop} 

\begin{proof}
For an instance of $X$ which  become a $X\lfrom U$, $X$ will necessary be produced by an elimination of $U$ (for the ``same" $X$). In this case, the$k$-extended-redex contains a smaller $k'$-ectended-redex. Only the $\odot_e$ rule allows to keep in the conclusion one of premises, it can be used an unspecified number of times without changing the concept of the same $X$.
\end{proof}

\subsection{Normalisation de L$_{\odot}$}

A \textbf{normal proof} is a proof which contains no $k$-extended-redexes and where all $\odot_e$ rules are as higher as possible.

Lets $\delta$ a proof, we define $PE(\delta)$ as the set of occurences of $\odot_e$ rules in $\delta$.
For $R\in PE(\delta)$, we define the two following integers:
\begin{enumerate}
\item the integer  $g(R)$ is the number of rules if there is a $k$-extended-redex in $B(S_0)$ and 0 otherwise with as conclusion premise $S_0$;

\item the integer $d_{conj}(R)$ is the number of rules do not belong to $PE(\delta)$ between $R$ and the rule which relies hypothesis $A$ and $B$ erased by $R$. 

\end{enumerate}

We define $h(\delta)$ as $\sum_{R \in PE(\delta)} d_{conj}(R)$
and $g(\delta)$ as $\min_{R\in PE(\delta)}(g(R))$ equal to $0$ is and only if $\delta$ no longer possesses $k$-extended-redex product.
We note $n(\delta)$ the number of rules of $\delta$.

In order to prove the normalization of L$_\odot$, we define the following measure over the proof $\delta$ be be a triplet of integers, which respect the lexicographic order:

$$|\delta|=\langle n(\delta), h(\delta), g(\delta) \rangle $$

\begin{propri}\label{propnormalisation}
A proof $\delta$ is normal if and only if it does not contain $0$-extended-redex and if $h(\delta) = 0$ and $g(\delta) = 0$.
\end{propri}

\begin{proof}
Let $\delta$ a proof of L$_\odot$, 
belongs to $|\delta|$:
\begin{itemize}
\item[$\circ$] the first integer is the classical one for Lambek calculus normalisation.
It is minimal if the proof does not contain $0$-extended-redex.

\item[$\circ$] the second integer give the process to rise $\odot_e$ to their higher position. In this phase, $k$-extended-redexes$_{\lfrom}$ and $_{\lto}$ appears and can be cancelled.
If every $\odot_e$ have their higher position in the proof then $h(\delta) = 0$.
Only $k$-extended-redexes $_{\odot}$ remains in $\delta$.
This case is presented in example 1 of Figure \ref{exformnorm}.

\item[$\circ$] the third integer represents the number of rules in a $k$-extended-redex$_{\odot}$.
When it is null, there is no more $k$-extended-redex${\odot}$ in $\delta$.
This case is presented in example 2 of Figure \ref{exformnorm}.
\end{itemize}
\end{proof}

\begin{figure}[htbp]
\begin{center}
$\begin{array}{l|l}
$example 1$& $example 2$\\
\prooftree
	\vdash C
	\prooftree
		\vdash E \odot F 
		\quad \prooftree
			C, E, F \vdash A \lfrom  B
			\justifies
			E, F \vdash C \lto (A \lfrom  B)
			\using
		             [\lto_{i}]
         			\thickness=0.07em
		\endprooftree
		\justifies
		\vdash C \lto (A \lfrom  B)
		\using
		   [\odot_{e}]
	\endprooftree
	\justifies
	\vdash  A \lfrom  B
	\using
	   [\lto_{e}]
\endprooftree
&
\prooftree
	\vdash E \odot F
	\prooftree
		\vdash C
		\quad \prooftree
			E \vdash (C \lto A) \lfrom  B
			\quad F \vdash B
			\justifies
			E, F \vdash C \lto A
			\using
		             [\lfrom _{e}]
         			\thickness=0.07em
		\endprooftree
		\justifies
		E, F \vdash A
		\using
		   [\lfrom _{e}]
	\endprooftree
	\justifies
	\vdash  A \lfrom  B
	\using
	   [\odot_{e}]
\endprooftree
\\
g(\lto_e) = 1 \Rightarrow& h(\odot_e) = 1\Rightarrow\\
k$-extended-redex$ & \odot_e$ does not have its position$
\end{array}$
\end{center}
\caption{Examples of proof which are not in normal form \label{exformnorm}}
\end{figure}

\begin{theo}
Every proof $\delta$ in L$_{\odot}$ calculus has a unique normal form.
\end{theo}

\begin{proof} We proceed by induction on $|\delta|$. 
By induction hypotheses, every proof $\delta '$ of size $|\delta'|< \langle r, d, g \rangle$ has a unique normal form.
Given a proof $\delta$ of size $= \langle r, d, g \rangle$, let us show that $\delta$ has a unique normal form as well.

$\bullet$ if $\delta$ has a redex, we could reduce it. 
The induced proof $\delta'$, $n(\delta') < n(\delta)$, hence $|\delta'| < \langle r, d, g \rangle$ and by induction, $\delta'$ has a unique normal form, then $\delta$ as well.

$\bullet$ Else:
\begin{verse}
If $d \neq 0$: let $R$ the lowest rule  $\odot_e$ such that $\neq0$. Hence, there exists a rule $R' \neq \odot_e$ higher than $R$, and $R$ can move upwards over all $\odot_e$ and over $R'$, next $R$ can rise over $R'$. the induced proof $\delta'$ is such that $n(\delta') = n(\delta) $ and $h(\delta') = h(\delta) - 1$.
Values of $d_{conj}(R_i)$, for $R_i$ $\odot_e$ rules below $R$, stands null because $R$ do not contribute to $d_{conj}(\_)$). 
Therefore $\delta' < \langle r, d, g \rangle$, and by induction , $\delta'$ has a unique normal form, then $\delta$ as well.

Else:

\begin{verse}
If $g\neq 0$: let $R'$ such that $g(R') = g$. This rule can move upwards above its left premise.  
The number of rules and the sum stand the same.
In its left part, we change a $\odot_e$ by a $\odot_e$) and $g$ decrease of 1. Hence the proof $\delta'$ is such that $|\delta'|<|\delta|$. By nduction, $\delta'$ has a unique normal form, then $\delta$ as well.

Else: using the property \ref{propnormalisation}, the proof is in normal form.
\end{verse}
\end{verse}

$\odot_e$ can appear only after any rule with two premises which links its two hypothesis $A$ and $B$.
Moreover, only one $\odot_e$ can be "as high as possible" after this rule use.
Indeed, a rule can connect only two hypotheses: the rightmost one of the left premise and the leftmost one  of the right premise..
We give a unique position to each $\odot_e$ thus the normal form is unique.

\end{proof}

All proofs have a unique normal form which could be performed using the strategy inside the proof above. 
Normal forms of the two previous examples, figure \ref{exformnorm} are the two following proofs:

\begin{center}
$\begin{array}{l|l}
$example 1$& $example 2$\\
\prooftree
	\vdash E \odot F 
	\prooftree
		\vdash C
		\quad \prooftree
			C, E, F \vdash A \lfrom  B
			\justifies
			E, F \vdash C \lto (A \lfrom  B)
			\using
		             [\lto_{i}]
         			\thickness=0.07em
		\endprooftree
		\justifies
		\vdash  A \lfrom  B
		\using
		   [\lto_{e}]
	\endprooftree
	\justifies
	\vdash C \lto (A \lfrom  B)
	\using
	   [\odot_{e}]
\endprooftree
&
\prooftree
	\vdash C
	\prooftree
		\vdash E \odot F
		\quad \prooftree
			E \vdash (C \lto A) \lfrom  B
			\quad F \vdash B
			\justifies
			E, F \vdash C \lto A
			\using
		             [\lfrom _{e}]
         			\thickness=0.07em
		\endprooftree
		\justifies
		\vdash  A \lfrom  B
		\using
		   [\odot_{e}]
	\endprooftree
	\justifies
	E, F \vdash A
	\using
	   [\lfrom _{e}]
\endprooftree
\\
g(\lto_e) = 0 & h(\odot_e) = 0\\
\end{array}$
\end{center}

\subsection{Sub-formula property for L$_{\odot}$}

\begin{theo}\label{prooftruc}
All proof of Lambek calculus with product on normal form $\delta$ of a sequent $\Gamma \vdash C$ holds the sub-formula property:
every formulae in a normal proof are sub-formula of some hypotheses $\Gamma$ or the conclusion of the proof $C$.
\end{theo}

\begin{proof} We proceed by induction :

We use a stronger definition of the sub-formula property:
every formulae in a normal sub-proof are formulae of some hypotheses or the conclusion of the proof and if the last rule used is an $\lto_e$ or $\lfrom _e$ every sub-formulae are sub-formulae of some hypotheses only.
Note that the axiom rule clearly satisfies the property of the sub-formula because the sequent is then one of the hypotheses.

We check the induction hypothesis after the use of each rule.

\begin{enumerate}
\item $\lto_e$:
let the proof $\delta$, where $\Gamma_i$ is the set of hypotheses used in the sub-proof $\delta_i$, for $i\in[2]$:
\begin{center}
\prooftree 
	\prooftree \Gamma_1 \leadsto 
	 C
	 \using \delta_1
	 \proofdotnumber=5
	\endprooftree 	
	  \prooftree \Gamma_2\leadsto 
	\prooftree 
		\justifies
		C\lto D
		\using
         		    [R]
         		\thickness=0.07em
	\endprooftree 
	\using \delta_2
	\proofdotnumber=5
	\endprooftree 
	\justifies
	D
	\using
             [\lto_e]
         \thickness=0.07em
\endprooftree 
\end{center}

Using the induction hypothesis:
\begin{itemize}
\item In $\delta_1$ every formulae are sub-formulae of $C$ or $\Gamma_1$;
\item In $\delta_2$ every formulae are sub-formulae of  $C\lto D$ or $\Gamma_2$.
\end{itemize}
The conclusion $D$ and the premise $C$ are direct sub-formulae of the premise $C \lto D$.
We have to check the rule $[R]$ above this premise:

\begin{itemize}
\item[$\circ$] if $R$ is $\lfrom _e$ or $\lto_e$: we use the induction hypothesis, we conclude that $C\lto D$ is a sub-formula of $\Gamma_2$.
Then every formula of $\delta$ is a sub-formula of $\Gamma_2$.

\item[$\circ$]  if $R$ is $\lto_i$: it is impossible because the rule should be a $0$-extended-redex, or $\delta$ is in normal form.
\item[$\circ$]  if $R$ is $\lfrom _i$: this case is structurally impossible because we could not derive $C \lto D$ with this rule.
\item[$\circ$]  if $R$ is $\odot_i$: this case is also impossible because we could not derive $C \lto D$ with this rule.
\item[$\circ$]  if $R$ is $\odot_e$. Once again, we must check the property with the rule $R'$ above:
\begin{center}
\prooftree 
	\prooftree \Gamma_1 \leadsto 
	 C
	 \using \delta_1
	 \proofdotnumber=5
	\endprooftree 	
	\quad 	\prooftree \Gamma_2[A,B]\leadsto 
	\prooftree 
		A\odot B
		\quad \prooftree 
			\justifies
			 C\lto D
			\using
         			   [R']
         		\thickness=0.07em
		\endprooftree 
		\justifies
		 C\lto D
		\using
         		   [\odot_e]
         	\thickness=0.07em
	\endprooftree 
	\using \delta_2
	\proofdotnumber=5
	\endprooftree 
	\justifies
	D
	\using
             [\lto_e]
         \thickness=0.07em
\endprooftree 

\bigskip
\end{center}

If $R'$ is $\lto_e$ or $\lfrom _e$, by the induction hypothesis, $C\lto D$ is a  sub-formula of some hypotheses.

If $R'$ is $\lto_i$: impossible because it should introduce a $1$-extended-redex, or $\delta$ is  in normal form... then it is not possible.

If $R'$ is one of the other introduction rules ($\lto_i$ or $\odot_i$): 
these case are structurally impossible. We could not derive $C \lto D$ with this rule.

If $R'$ is  $\odot_e$, once again, we check the property on the rule above
Remark that the number of rules above this rule is finite and they constitute a sequence such that:

 \begin{center}
\prooftree 
	\prooftree \Gamma_1 \leadsto 
	 C
	 \using \delta_1
	 \proofdotnumber=5
	\endprooftree 	
	\quad 	
	\prooftree 
		A_1\odot B_1
		\prooftree 
			A_n\odot B_n
			\prooftree 
			\prooftree  \Gamma_2[A_1, \cdots, A_n,B1, \cdots, B_n]\leadsto 
			\using \delta_2
			\endprooftree 
			\justifies
			 C\lto D
			\using
         			   [R]
         			\thickness=0.07em
		\endprooftree 
		\justifies
		\prooftree 
		\leadsto
		 C\lto D
		\endprooftree 
		\using
         		   [\odot_e]
         		\thickness=0.07em
	\endprooftree 
	\justifies
	 C\lto D
	\using
         	  [\odot_e]
         	\thickness=0.07em
	\endprooftree 
	\justifies
	D
	\using
             [\lto_e]
         \thickness=0.07em
\endprooftree 
\end{center}

In this case, we have:
\begin{itemize}

\item either there is only $\odot_e$ in this sequence then $C \lto D$ is one of the hypotheses.
\item either it exist a rules $R''$ different than $\odot_e$ inside then using the argument of the corresponding case above we prove that $C\lto D$ is sub-formula oh the hypotheses.
\end{itemize}
\end{itemize}
In every case, the conclusion of $\lto_e$ is a sub-formula of the hypotheses. The induction hypothesis is checked.

\item $\lfrom _e$:
let the proof $\delta$, where $\Gamma_i$ is the set of hypotheses used in the sub-proof $\delta_i$, for $i\in[2]$:

\begin{center}
\prooftree 	
	\prooftree 
		\Gamma_2 
		\leadsto 
		\prooftree 
			\justifies
			D\lfrom  C
			\using
         			    [R]
         			\thickness=0.07em
		\endprooftree 
		\using \delta_2
		\proofdotnumber=5
	\endprooftree 
	\quad \prooftree \Gamma_1 \leadsto 
	 C
	 \using \delta_1
	\proofdotnumber=5
	\endprooftree 
	\justifies
	D
	\using
             [\lto_e]
         \thickness=0.07em
\endprooftree 
\end{center}

The conclusion $D$ and the premise $C$ are direct sub-formulae of the premise $D \lfrom  C$.
We have to check the rule $[R]$ above this premise whose conclusion is  $D \lfrom  C$:
this case is the same as $\lto_e$. 
In the same way, we prove that $D\lfrom C$ is sub-formula of  $\Gamma_2$.

\item $\lto_i$:
let the proof $\delta$,  where $\Gamma_1$ is the set of hypotheses used in the sub-proof $\delta_1$:
\begin{center}
\prooftree 
	\prooftree C, \Gamma_1 \leadsto 
	 	D
	 	\using \delta_1
	 	\proofdotnumber=5
	\endprooftree 	
	\justifies
	C \lto D
	\using
             [\lto_i]
         \thickness=0.07em
\endprooftree 
\end{center}

In $\delta_1$ every formula is a sub-formula of $D$ or of $C$ and $\Gamma_1$.
Furthermore, $D$ is a sub-formula of $C\lto D$.
Then, every formula of $\delta$ is sub-formula of $C,\Gamma_1$ or $C\lto D$.

\item $\lfrom _i$:
let the proof $\delta$,  where $\Gamma_1$ is the set of hypotheses used in the sub-proof $\delta_1$:
\begin{center}
\prooftree 
	\prooftree C, \Gamma_1 \leadsto 
	 D
	 \using \delta_1
 	\proofdotnumber=5
	\endprooftree 	
	\justifies
	D \lfrom  C
	\using
             [\lfrom _i]
         \thickness=0.07em
\endprooftree 
\end{center}

This case is strictly symmetrical to that of $\lto_i$: $D$ is a sub-formule of $D \lfrom  C$, and every formula of $\delta_1$ is sub-formula of $C, \Gamma_1$ or $D$. By trnasitivity, any formula of $\delta$ is sub-formula of $C, \Gamma$ or $D\lfrom C$.

\item $\odot_i$: let the proof $\delta$, where $\Gamma_i$ is the set of hypotheses used in the sub-proof $\delta_i$, for $i\in[2]$:

\begin{center}
\prooftree 
	\prooftree \Gamma_1 \leadsto 
		 C
		 \using \delta_1
	 	\proofdotnumber=5
	\endprooftree 	
	\quad \prooftree \Gamma_2 \leadsto 
	 	D
	 	\using \delta_2
	 	\proofdotnumber=5
	\endprooftree 	
	\justifies
	C \odot D
	\using
             [\odot_i]
         \thickness=0.07em
\endprooftree 
\end{center}

\begin{itemize}
\item In $\delta_1$ every formula is a sub-formula of $C$ or of $\Gamma_1$.
\item In $\delta_2$ every formula is a sub-formula of $D$ or of $\Gamma_2$.
\end{itemize}

Futhermore , $C$ and $D$ are sub-formulae of $C\odot D$. 
By transitivity, every formula of $\delta$ is sub-formula of $\Gamma_1, \Gamma_2$ or of $C\odot D$.

\item $\odot_e$: let the proof $\delta$, where $\Gamma_i$ is the set of hypotheses used in the sub-proof $\delta_i$, for $i\in[2]$:
\begin{center}
\prooftree 
	\prooftree \Gamma_1 \leadsto 
		 A \odot B
		 \using \delta_1
	 	\proofdotnumber=5
	\endprooftree 	
	\quad \prooftree \Gamma_2 \leadsto 
	 	D
		 \using \delta_2
	 	\proofdotnumber=5
	\endprooftree 	
	\justifies
	D
	\using
             [\odot_e]
         \thickness=0.07em
\endprooftree 
\end{center}

\begin{itemize}
\item In $\delta_1$, every formula is a sub-formula of $A\odot B$ or of $\Gamma_1$.
\item In $\delta_2$, every formula is a sub-formula of $D$ or of $\Gamma_2$.
\end{itemize}

The conclusion of $\delta$ is the conclusion of one premise, then the property stands for the part of the proof which the conclusion belongs, \textit{i.e.} $\delta_2$.
We check the property for the second part of the proof. here, we have to check that $A \odot B$ is sub-formula of hypothesis of $\delta_1$. For this, we analyze the rule $R$ above:
\begin{itemize}
\item[$\circ$] if $R$ is $\lto_e$ ou $\lfrom _e$, by induction and th type of the rule, $A\odot B$ is sub-formula of de $\Gamma_1$
\item[$\circ$] if $R$ is $\lto_i$ ou $\lfrom _i$: this case is structurally impossible because $A \odot B$ could not be produced after these lasts.
\item[$\circ$] if $R$ is $\odot_i$: this case is absurd because the succession of these two rules imply a $0$-extended-redex, but the proof is in normal form.
\item[$\circ$] if $R$ is $\odot_e$. We check the property for the rule $R$ above: 

$\;$
\begin{center}
\prooftree 
	
	 \prooftree 
	 	E \odot F
		\quad
			\prooftree  \Gamma_1[E,F] \leadsto  A \odot B
				\using \delta_1
			 	\proofdotnumber=5
			\endprooftree 
		\justifies
	 	A \odot B
		\using
	             [\odot_e]
	         \thickness=0.07em
	\endprooftree 
	\quad \prooftree \Gamma_2[A,B] \leadsto  D
		\using \delta_2
	 	\proofdotnumber=5
	\endprooftree 
	\justifies
	D
	\using
             [\odot_e]
         \thickness=0.07em
\endprooftree 
\end{center}

\bigskip

\begin{itemize}
\item[$\bullet$] If $A\odot B$ is not an hypothesis: in this case, a conjoined rule $\odot_i$ exists for the analyzed rule. 
This implies that a $k$-extended-redex is in the proof, but we assume that it is in normal form

\item[$\bullet$] If $A\odot B$ is an hypothesis of $\Gamma_1$, then $A\odot B$ is a sub-formula of hypothesis.
\end{itemize}
\end{itemize}

In any possible case, $A\odot B$ is sub-formula of hypothesis.
Thereby, the sub-formula property stands for $\odot_e$.
\end{enumerate}
\end{proof}

In L$_{\odot}$, every proof have a unique normal form which check the sub-formula property.
We observe that unlike \cite{Negri02nd}, rules use are the usual one for this calculus.

\section{Normalization of proofs of PCMLL}

Now, we present a normalization for proofs of PCMLL, from which we prove that the sub-formula property holds.
In the same way for L$_{\odot}$, the normalization uniquely positioned eliminations of non-commutative product, and build sequence of commutative product eliminations.
The relative position of a commutative product removal in a sequence is not unique

\subsection{Property of PCMLL}

\begin{propri}[product eliminations could rise in the proof]  \label{product-upwards}

Let $R$ a product elimination $\otimes_e$ (\resp a rule $\odot_e$) of $\Gamma[\Delta]\seq C$ between a proof $\delta_0$  $\Delta\seq A\otimes B$ and a proof with conclusion $\Gamma[(A,B)]\seq C$ (\resp $\Gamma[\langle A;B\rangle]\seq C$) obtained by a rule $R'$ of a proof $\delta_1$ of a sequent $\Theta[(A,B)]\seq X$ 
(\resp $\Theta[\langle A;B\rangle]\seq X$) (and if $R'$ is a binary rule of a second proof $\delta_2$ with $\Psi\seq U$ as conclusion). 

Then, we can obtain a proof for the same sequent $\Gamma[\Delta]\seq C$ by first applying the rule $\otimes_e$ (\resp the rule $\odot_e$) between the proof $\delta_0$ with $\Delta\seq A\otimes B$ as conclusion and the proof $\delta_1$ with $\Theta[(A,B)]\seq X$ as conlusion  (\resp $\Theta[\langle A;B\rangle]\seq X$) giving the sequent $\Theta[\Delta]\seq X$. By applying the rule $R'$ to this new proof and  potentially the proof $\delta_2$, we get the same sequent $\Gamma[\delta]\seq C$. 

This derivation is presented in figure \ref{montprod}
\end{propri}

\begin{figure}[htbp]
\begin{center}

\prooftree
	\prooftree
		\leadsto
		\Delta \seq A\otimes B
		\using
		  \delta_0
	\endprooftree
	\quad 
	\prooftree
		\prooftree
			\leadsto
			\Psi \seq U
			\using
			  \delta_2			
		\endprooftree
		\prooftree
			\leadsto
			\Theta[(A,B)] \seq X
			\using
			  \delta_1			
		\endprooftree
		\justifies
		\Gamma[(A,B)] \seq C
		\using R'
	\endprooftree
	\justifies
	\Gamma[\Delta] \seq C
	\using R
\endprooftree
$\quad \Rightarrow \quad$
\prooftree
	\prooftree
		\leadsto
		\Psi \seq U
		\using
		  \delta_2			
	\endprooftree
	\prooftree
		\prooftree
			\leadsto
			\Delta \seq A\otimes B
			\using
			  \delta_0
		\endprooftree
		\quad 
		\prooftree
			\leadsto
			\Theta[(A,B)] \seq X
			\using
			  \delta_1			
		\endprooftree	
		\justifies
		\Theta[\Delta] \seq X
		\using R
	\endprooftree
	\justifies
	\Gamma[\Delta] \seq C
	\using R'
\endprooftree
	
\end{center}
\caption{Proof's structure allowing the rise of the product in PCMLL. \label{montprod}}
\end{figure}

\begin{proof}

The proof is similar to property \ref{proofriselambek}.
This is a case study according to rule over the product elimination.
This elimination could only rise when hypothesis which must be cancelled are in the same premise and hold their respective position according order required by the elimination.

Let check every cases for $\otimes_e$.

\begin{itemize}
\item[$\circ$] Rise over $\lto_e$.
\begin{itemize}
\item[$\bullet$] if hypothesis are in the left premise of $\lto_e$:

\bigskip
 \prooftree 
 	\Delta \vdash A \otimes B
	\quad \prooftree 
		\Gamma[(A,B)] \vdash D
		\quad \Phi \vdash  D \lto C
		\justifies
		<\Gamma[(A,B)] ; \Phi> \vdash C 
		\using
	             [\lto_{e}]
	         \thickness=0.07em
	\endprooftree 
	\justifies
	<\Gamma[\Delta]; \Phi> \vdash C 
	\using
             [\otimes_{e}]
         \thickness=0.07em
\endprooftree 

\begin{flushright}
$\Rightarrow$
\prooftree 
	 \prooftree 
 		\Delta \vdash A \otimes B
		\quad \Gamma[(A,B)] \vdash D
		\justifies
		\Gamma[\Delta] \vdash D
		\using
	             [\otimes_{e}]
	         \thickness=0.07em
	\endprooftree 
	\Phi \vdash  D \lto C
	\justifies
	<\Gamma[\Delta]; \Phi> \vdash C 
	\using
             [\lto_{e}]
         \thickness=0.07em
\endprooftree 
\end{flushright}
\bigskip
\item[$\bullet$] if hypothesis are in the right premise of $\lto_e$:\\
\bigskip
 \prooftree 
 	\Delta \vdash A \otimes B
	\quad \prooftree 
		\Gamma \vdash D
		\quad \Phi[(A,B)] \vdash  D \lto C
		\justifies
		<\Gamma ; \Phi[(A,B)]> \vdash C 
		\using
	             [\lto_{e}]
	         \thickness=0.07em
	\endprooftree 
	\justifies
	<\Gamma; \Phi[\Delta]> \vdash C 
	\using
             [\otimes_{e}]
         \thickness=0.07em
\endprooftree 

\begin{flushright}
$\Rightarrow$
\prooftree 
	\quad \Gamma \vdash  D
	 \prooftree 
 		\Delta \vdash A \otimes B
		\quad \Phi[(A,B)] \vdash  D \lto C
		\justifies
		\Phi[\Delta] \vdash D \lto C
		\using
	             [\otimes_{e}]
	         \thickness=0.07em
	\endprooftree 
	\justifies
	<\Gamma; \Phi[\Delta]> \vdash C 
	\using
             [\lto_{e}]
         \thickness=0.07em
\endprooftree
\end{flushright}

\end{itemize}
\bigskip

\item[$\circ$] Rise over $\lfrom _e$.
\begin{itemize}

\item[$\bullet$] if hypothesis are in the right premise of $\lfrom _e$:
\bigskip

 \prooftree 
 	\Delta \vdash A \otimes B
	\quad \prooftree 
		\Phi \vdash  C \lfrom  D
		\quad \Gamma[(A,B)] \vdash D
		\justifies
		<\Phi ; \Gamma[(A,B)] > \vdash C 
		\using
	             [\lfrom _{e}]
	         \thickness=0.07em
	\endprooftree 
	\justifies
	<\Phi ; \Gamma[\Delta]> \vdash C 
	\using
             [\otimes_{e}]
         \thickness=0.07em
\endprooftree
\begin{flushright}
$\Rightarrow$
\prooftree 
	\Phi \vdash C \lfrom  D
	\quad  \prooftree 
 		\Delta \vdash A \otimes B
		\quad \Gamma[(A,B)] \vdash D
		\justifies
		\Gamma[\Delta] \vdash D
		\using
	             [\otimes_{e}]
	         \thickness=0.07em
	\endprooftree 
	\justifies
	< \Phi ; \Gamma[\Delta]> \vdash C 
	\using
             [\lfrom _{e}]
         \thickness=0.07em
\endprooftree 
\end{flushright}
\bigskip
\item[$\bullet$] if hypothesis are in the left premise of $\lfrom _e$:

\bigskip
 \prooftree 
 	\Delta \vdash A \otimes B
	\quad \prooftree 
		\Phi[(A,B)] \vdash C \lfrom  D
		\quad \Gamma \vdash D
		\justifies
		<\Phi[(A,B)] ; \Gamma> \vdash C 
		\using
	             [\lfrom _{e}]
	         \thickness=0.07em
	\endprooftree 
	\justifies
	<\Phi[\Delta] ; \Gamma> \vdash C 
	\using
             [\otimes_{e}]
         \thickness=0.07em
\endprooftree 

\begin{flushright}
$\Rightarrow$
\prooftree 
	\prooftree 
 		\Delta \vdash A \otimes B
		\quad \Phi[(A,B)] \vdash C \lfrom  D
		\justifies
		\Phi[\Delta] \vdash C \lfrom  D
		\using
	             [\otimes_{e}]
	         \thickness=0.07em
	\endprooftree 
	\Gamma \vdash  D
	\justifies
	<\Phi[\Delta] ; \Gamma> \vdash C 
	\using
             [\lfrom _{e}]
         \thickness=0.07em
\endprooftree 
\end{flushright}
\end{itemize}
\bigskip

\item[$\circ$] Rise over $\llto_e$

\bigskip
\begin{itemize}
\item[$\bullet$] if hypothesis are in the left premise of $\llto_e$:

\bigskip
 \prooftree 
 	\Delta \vdash A \otimes B
	\quad \prooftree 
		\Gamma[(A,B)] \vdash D
		\quad \Phi \vdash  D \llto C
		\justifies
		(\Gamma[(A,B)] , \Phi) \vdash C 
		\using
	             [\llto_{e}]
	         \thickness=0.07em
	\endprooftree 
	\justifies
	(\Gamma[\Delta]; \Phi) \vdash C 
	\using
             [\otimes_{e}]
         \thickness=0.07em
\endprooftree 

\begin{flushright}
$\Rightarrow$
\prooftree 
	 \prooftree 
 		\Delta \vdash A \otimes B
		\quad \Gamma[(A,B)] \vdash D
		\justifies
		\Gamma[\Delta] \vdash D
		\using
	             [\otimes_{e}]
	         \thickness=0.07em
	\endprooftree 
	\Phi \vdash  D \llto C
	\justifies
	(\Gamma[\Delta], \Phi) \vdash C 
	\using
	   [\llto_{e}]
         \thickness=0.07em
\endprooftree 
\end{flushright}
\bigskip

\item[$\bullet$] if hypothesis are in the right premise of $\llto_e$:

\bigskip
 \prooftree 
 	\Delta \vdash A \otimes B
	\quad \prooftree 
		\Gamma \vdash D
		\quad \Phi[(A,B)] \vdash  D \llto C
		\justifies
		(\Gamma , \Phi[(A,B)]) \vdash C 
		\using
	             [\llto_{e}]
	         \thickness=0.07em
	\endprooftree 
	\justifies
	(\Gamma , \Phi[\Delta]) \vdash C 
	\using
             [\otimes_{e}]
         \thickness=0.07em
\endprooftree 

\begin{flushright}
$\Rightarrow$
\prooftree 
	\quad \Gamma \vdash  D
	 \prooftree 
 		\Delta \vdash A \otimes B
		\quad \Phi[(A,B)] \vdash  D \llto C
		\justifies
		\Phi[\Delta] \vdash D \llto C
		\using
	             [\otimes_{e}]
	         \thickness=0.07em
	\endprooftree 
	\justifies
	(\Gamma, \Phi[\Delta]) \vdash C 
	\using
             [\llto_{e}]
         \thickness=0.07em
\endprooftree 
\end{flushright}
\end{itemize}
\bigskip

\item[$\circ$] Rise over $\lfrom _i$:

\bigskip
 \prooftree 
 	\Delta \vdash A \otimes B
	\quad \prooftree 
		<\Gamma[(A,B)] ; D> \vdash C
		\justifies
		\Gamma[(A,B)] \vdash C \lfrom  D
		\using
	             [\lfrom _{i}]
	         \thickness=0.07em
	\endprooftree 
	\justifies
	\Gamma[\Delta] \vdash C \lfrom  D
	\using
             [\otimes_{e}]
         \thickness=0.07em
\endprooftree 
$\Rightarrow$
\prooftree 
	 \prooftree 
 		\Delta \vdash A \otimes B
		\quad <\Gamma[(A,B)] ; D> \vdash C
		\justifies
		<\Gamma[\Delta ] ; D> \vdash C
		\using
	             [\otimes_{e}]
	         \thickness=0.07em
	\endprooftree 
	\justifies
	\Gamma[\Delta] \vdash C \lfrom  D
	\using
	   [\lfrom _{i}]
         \thickness=0.07em
\endprooftree 

\bigskip

\item[$\circ$] Rise over $\lto_i$: $\;$

\bigskip
 \prooftree 
 	\Delta \vdash A \otimes B
	\quad \prooftree 
		<D ; \Gamma[(A,B)] > \vdash C
		\justifies
		\Gamma[(A,B)] \vdash D \lto C
		\using
	             [\lto_{i}]
	         \thickness=0.07em
	\endprooftree 
	\justifies
	\Gamma[\Delta] \vdash D \lto C 
	\using
             [\otimes_{e}]
         \thickness=0.07em
\endprooftree 
$\Rightarrow$
\prooftree 
	 \prooftree 
 		\Delta \vdash A \otimes B
		\quad <D ; \Gamma[(A,B)] > \vdash C
		\justifies
		<D ; \Gamma[\Delta ]> \vdash C
		\using
	             [\lto_{i}]
	         \thickness=0.07em
	\endprooftree 
	\justifies
	\Gamma[\Delta] \vdash D \lto C
	\using
	   [\lto_{i}]
         \thickness=0.07em
\endprooftree 

\bigskip

\item[$\circ$] Rise over $\llto_i$:

\bigskip
 \prooftree 
 	\Delta \vdash A \otimes B
	\quad \prooftree 
		(\Gamma[(A,B)] , D) \vdash C
		\justifies
		\Gamma[(A,B)] \vdash D \llto C
		\using
	             [\llto_{i}]
	         \thickness=0.07em
	\endprooftree 
	\justifies
	\Gamma[\Delta] \vdash D \llto C 
	\using
             [\otimes_{e}]
         \thickness=0.07em
\endprooftree 
$\Rightarrow$
\prooftree 
	 \prooftree 
 		\Delta \vdash A \otimes B
		\quad (\Gamma[(A,B)] , D) \vdash C
		\justifies
		(\Gamma[\Delta ] , D) \vdash C
		\using
	             [\otimes_{e}]
	         \thickness=0.07em
	\endprooftree 
	\justifies
	\Gamma[\Delta] \vdash D \llto C
	\using
	   [\llto_{i}]
         \thickness=0.07em
\endprooftree 

\bigskip

\item[$\circ$] Rise over $\otimes_e$:

\bigskip
\begin{itemize}
\item[$\bullet$] if hypothesis are in the right premise of $\otimes_e$:

\bigskip
\prooftree 
	\Gamma \vdash A \otimes B
	\prooftree 
		\Delta \vdash C \otimes D
		\quad (\Phi, (A,B),(C,D), \Phi ' )\vdash E
		\justifies
		(\Phi, (A,B),\Delta, \Phi ') \vdash E
		\using
	             [\otimes_{e}]
	         \thickness=0.07em
	\endprooftree 
	\justifies
	(\Phi, \Gamma,\Delta, \Phi ' ) \vdash E
	\using
             [\otimes_{e}]
         \thickness=0.07em
\endprooftree 

\begin{flushright}
$\Rightarrow$
\prooftree 
	\Delta \vdash C \otimes D
	\prooftree 
		\Gamma \vdash A \otimes B
		\quad (\Phi, (A,B),(C,D), \Phi ' ) \vdash E
		\justifies
		(\Phi, \Gamma,(C,D), \Phi ') \vdash E
		\using
	             [\otimes_{e}]
	         \thickness=0.07em
	\endprooftree 
	\justifies
	(\Phi, \Gamma,\Delta, \Phi ' ) \vdash E
	\using
             [\otimes_{e}]
         \thickness=0.07em
\endprooftree 
\end{flushright}
\bigskip

\item[$\bullet$] if hypothesis are in the left premise of $\otimes_e$:

\bigskip
\prooftree 
	\Gamma \vdash A \otimes B
	\prooftree 
		(\Delta, (A,B), \Delta ' )\vdash C \otimes D
		\quad (\Phi, (C,D), \Phi ' )\vdash E
		\justifies
		(\Phi, \Delta, (A,B), \Delta ', \Phi ') \vdash E
		\using
	             [\otimes_{e}]
	         \thickness=0.07em
	\endprooftree 
	\justifies
	(\Phi, \Delta, \Gamma, \Delta ', \Phi ' ) \vdash E
	\using
             [\otimes_{e}]
         \thickness=0.07em
\endprooftree 
\bigskip

\begin{flushright}
$\Rightarrow$
\prooftree 
	\prooftree 
		\Gamma \vdash A \otimes B
		\quad (\Delta, (A,B), \Delta ' ) \vdash C \otimes D
		\justifies
		(\Delta, \Gamma, \Delta ' ) \vdash C \otimes D
		\using
	             [\otimes_{e}]
	         \thickness=0.07em
	\endprooftree 
	(\Phi, (C,D), \Phi ') \vdash E
	\justifies
	(\Phi, \Delta, \Gamma, \Delta, \Phi ' ) \vdash E
	\using
             [\otimes_{e}]
         \thickness=0.07em
\endprooftree 
\end{flushright}
\end{itemize}
\bigskip

\item[$\circ$] Rise over $\otimes_i$:
\begin{itemize}
\item[$\bullet$] if hypothesis are in the left premise of $\otimes_i$:

\bigskip
\prooftree 
	\Gamma \vdash A \otimes B
	\prooftree 
		(\Delta, (A,B), \Delta ' ) \vdash C
		\quad \Phi \vdash D
		\justifies
		(\Delta, (A,B), \Delta ', \Phi) \vdash C \otimes D
		\using
	             [\otimes_{i}]
	         \thickness=0.07em
	\endprooftree 
	\justifies
	(\Delta, \Gamma, \Delta ', \Phi )\vdash C \otimes D
	\using
             [\otimes_{e}]
         \thickness=0.07em
\endprooftree 

\begin{flushright}
$\Rightarrow$
\prooftree 
	\prooftree 
		\Gamma \vdash A \otimes B
		\quad (\Delta, (A,B), \Delta ' ) \vdash C
		\justifies
		(\Delta, \Gamma, \Delta ' ) \vdash C
		\using
	             [\otimes_{e}]
	         \thickness=0.07em
	\endprooftree 
	\quad \Phi \vdash D
	\justifies
	(\Delta, \Gamma , \Delta ', \Phi ) \vdash C \otimes D
	\using
            [\otimes_{i}]
         \thickness=0.07em
\endprooftree 
\end{flushright}
\bigskip

\item[$\bullet$] if hypothesis are in the right premise of $\otimes_i$:

\bigskip
\prooftree 
	\Gamma \vdash A \otimes B
	\prooftree 
		\Delta \vdash C
		\quad (\Phi, (A,B), \Phi ') \vdash D
		\justifies
		(\Delta, \Phi, (A,B) , \Phi ') \vdash C \otimes D
		\using
	             [\otimes_{i}]
	         \thickness=0.07em
	\endprooftree 
	\justifies
	(\Delta, \Phi, \Gamma, \Phi ' ) \vdash C \otimes D
	\using
             [\otimes_{e}]
         \thickness=0.07em
\endprooftree 

\begin{flushright}
$\Rightarrow$
\prooftree 
	\Delta \vdash C
	\quad \prooftree 
		\Gamma \vdash A \otimes B
		\quad  (\Phi, (A,B), \Phi ') \vdash D
		\justifies
		(\Phi, \Gamma, \Phi ' ) \vdash D
		\using
	             [\otimes_{e}]
	         \thickness=0.07em
	\endprooftree 
	\justifies
	(\Delta, \Phi, \Gamma, \Phi ' ) \vdash C \otimes D
	\using
            [\otimes_{i}]
         \thickness=0.07em
\endprooftree 
\end{flushright}
\end{itemize}
\bigskip

\item[$\circ$] Rise over $\odot_e$:
\begin{itemize}
\item[$\bullet$] if hypothesis are in the right premise of$\odot_e$: 

\bigskip
\prooftree 
	\Gamma \vdash A \otimes B
	\prooftree 
		\Delta \vdash C \odot D
		\quad (\Phi, (A,B),\Psi, <C;D>, \Psi ', \Phi ' )\vdash E
		\justifies
		(\Phi, (A,B), \Psi, \Delta, \Psi ', \Phi ') \vdash E
		\using
	             [\odot_{e}]
	         \thickness=0.07em
	\endprooftree 
	\justifies
	(\Phi, \Gamma, \Psi, \Delta, \Psi ', \Phi ' ) \vdash E
	\using
             [\otimes_{e}]
         \thickness=0.07em
\endprooftree 

\begin{flushright}
$\Rightarrow$
\prooftree 
	\Delta \vdash C \odot D
	\prooftree 
		\Gamma \vdash A \otimes B
		\quad (\Phi, (A,B), \Psi, <C;D>, \Psi ', \Phi ' ) \vdash E
		\justifies
		(\Phi, \Gamma,\Psi, <C;D>, \Psi ', \Phi ') \vdash E
		\using
	             [\otimes_{e}]
	         \thickness=0.07em
	\endprooftree 
	\justifies
	(\Phi, \Gamma,\Psi, \Delta, \Psi ' \Phi ' ) \vdash E
	\using
             [\odot_{e}]
         \thickness=0.07em
\endprooftree 
\end{flushright}
\bigskip

\item[$\bullet$] if hypothesis are in the left premise of $\odot_e$:

\bigskip
\prooftree 
	\Gamma \vdash A \otimes B
	\prooftree 
		(\Delta, (A,B), \Delta ' )\vdash C \odot D
		\quad (\Phi, \Psi, <C;D>, \Psi ', \Phi ' )\vdash E
		\justifies
		(\Phi, \Psi, \Delta,  (A,B) , \Delta ',  \Psi ', \Phi ') \vdash E
		\using
	             [\odot_{e}]
	         \thickness=0.07em
	\endprooftree 
	\justifies
	(\Phi, \Psi, \Delta,  \Gamma, \Delta ', \Psi ',\Phi ' ) \vdash E
	\using
             [\otimes_{e}]
         \thickness=0.07em
\endprooftree 
\bigskip 

\begin{flushright}
$\Rightarrow$
\prooftree 
	\prooftree 
		\Gamma \vdash A \otimes B
		\quad (\Delta, (A,B), \Delta ' ) \vdash C \odot D
		\justifies
		(\Delta, \Gamma, \Delta ' ) \vdash C \odot D
		\using
	             [\otimes_{e}]
	         \thickness=0.07em
	\endprooftree 
	(\Phi, \Psi, <C;D>, \Psi ', \Phi ') \vdash E
	\justifies
	(\Phi, \Psi, \Delta,  \Gamma, \Delta' , \Psi ', \Phi ' ) \vdash E
	\using
             [\odot_{e}]
         \thickness=0.07em
\endprooftree 
\end{flushright}
\end{itemize}
\bigskip

\item[$\circ$] Rise over $\odot_i$:
\begin{itemize}
\item[$\bullet$] if hypothesis are in the left premise of $\odot_i$:

\bigskip
\prooftree 
	\Gamma \vdash A \otimes B
	\quad \prooftree 
		(\Delta, (A,B), \Delta ' ) \vdash C
		\quad \Phi \vdash D
		\justifies
		<(\Delta, (A,B), \Delta '); \Phi> \vdash C \odot D
		\using
	             [\odot_{i}]
	         \thickness=0.07em
	\endprooftree 
	\justifies
	<(\Delta, \Gamma, \Delta '); \Phi >\vdash C \odot D
	\using
             [\otimes_{e}]
         \thickness=0.07em
\endprooftree 

\begin{flushright}
$\Rightarrow$
\prooftree 
	\prooftree 
		\Gamma \vdash A \otimes B
		\quad (\Delta, (A,B) , \Delta ' ) \vdash C
		\justifies
		(\Delta, \Gamma, \Delta ' ) \vdash C
		\using
	             [\otimes_{e}]
	         \thickness=0.07em
	\endprooftree 
	\quad \Phi \vdash D
	\justifies
	<(\Delta, \Gamma , \Delta ') ; \Phi > \vdash C \odot D
	\using
            [\odot_{i}]
         \thickness=0.07em
\endprooftree 
\end{flushright}
\bigskip

\item[$\bullet$] if hypothesis are in the right premise of $\odot_i$:

\bigskip
\prooftree 
	\Gamma \vdash A \otimes B
	\prooftree 
		\Delta \vdash C
		\quad (\Phi, (A,B), \Phi ') \vdash D
		\justifies
		<\Delta ; (\Phi, (A,B), \Phi ')> \vdash C \odot D
		\using
	             [\odot_{i}]
	         \thickness=0.07em
	\endprooftree 
	\justifies
	<\Delta ; (\Phi, \Gamma, \Phi ' ) > \vdash C \odot D
	\using
             [\otimes_{e}]
         \thickness=0.07em
\endprooftree 

\begin{flushright}
$\Rightarrow$
\prooftree 
	\Delta \vdash C
	\quad \prooftree 
		\Gamma \vdash A \otimes B
		\quad  (\Phi, (A,B), \Phi ') \vdash D
		\justifies
		(\Phi, \Gamma, \Phi ' ) \vdash D
		\using
	             [\otimes_{e}]
	         \thickness=0.07em
	\endprooftree 
	\justifies
	<\Delta ;  (\Phi, \Gamma, \Phi ' )> \vdash C \odot D
	\using
            [\odot_{i}]
         \thickness=0.07em
\endprooftree 
\end{flushright}
\bigskip
\end{itemize}
\item[$\circ$] Rise over $\sqsubset$:

\bigskip
\begin{flushright}
\prooftree 
	\prooftree 
		\Gamma \vdash A \otimes B
		\justifies
		\Gamma ' \vdash A \otimes B
		\using
	             [\sqsubset]
	         \thickness=0.07em
	\endprooftree 
	\quad 
	\Delta[A,B] \vdash D
	\justifies
	\Delta [\Gamma ']  \vdash D
	\using 
             [\otimes_{e}]
         \thickness=0.07em
\endprooftree 
$\quad\Rightarrow\quad$
\prooftree 
	\prooftree 
		\Gamma \vdash A \otimes B
		\quad  \Delta[A,B] \vdash D
		\justifies
		\Delta[\Gamma] \vdash D
		\using
	             [\otimes_{e}]
	         \thickness=0.07em
	\endprooftree 
	\justifies
	\Delta[\Gamma '] \vdash D
	\using
            [\sqsubset]
         \thickness=0.07em
\endprooftree 
\end{flushright}

(the removal of hypothesis in $\Gamma$ does not modify the order)

\end{itemize}
The check of the property for non-commutativity is an extension of the property \ref{proofriselambek} which is similar to the previous case.
\end{proof}

The procedure is analogous to the case studied for L$_\odot$. 
To do this we introduce \textbf{redexes} of the calculus.
Mix logic contains seven redexes: one for each implicative connective and two for each product connective, the conjoined introduction could be in the left premise or in the right premise.

The following present the seven redexes:
\begin{enumerate}
\item[$\circ$] Redex$_{\lfrom }$: introduction $\lfrom _i$ and direct elimination of $\lfrom _e$.
\begin{center}
\prooftree 
	\prooftree 
		\prooftree \leadsto  \langle \Gamma; D\rangle \vdash C \endprooftree 
		\justifies
		\Gamma \vdash C \lfrom  D 
		\using
	         	    [\lfrom _{i}]
         		\thickness=0.07em
	\endprooftree 
	\quad \prooftree \leadsto  \Delta \vdash D \using \delta_1 \proofdotnumber=5 \endprooftree 
	\justifies
	\langle \Gamma; \Delta \rangle \vdash C
	\using
             [\lfrom _{e}]
         \thickness=0.07em
\endprooftree 
$\quad \Rightarrow\quad$
\prooftree 
	\prooftree   
		\leadsto \Delta \vdash D 
		\using \delta_1\proofdotnumber=5
	\endprooftree  
	\leadsto  \langle \Gamma ; \Delta \rangle \vdash C
\endprooftree
\end{center}

\item[$\circ$] Redex$_{\lto}$: introduction  $\lto_i$ and direct elimination of $\lto_e$.
\begin{center}
\prooftree 
	\prooftree \leadsto \Delta \vdash D \using \delta_1 \proofdotnumber=5 \endprooftree 
	\quad\prooftree 
		\prooftree \leadsto \langle D ; \Gamma \rangle\vdash  C\endprooftree 
		\justifies
		\Gamma \vdash D \lto C 
		\using
	         	    [\lto_{i}]
         		\thickness=0.07em
	\endprooftree 
	\justifies
	\langle \Delta ; \Gamma \rangle \vdash C
	\using
             [\lto_{e}]
         \thickness=0.07em
\endprooftree 
$\quad \Rightarrow\quad$
\prooftree \prooftree   \leadsto \Delta \vdash D \using \delta_1 \proofdotnumber=5\endprooftree \leadsto  \langle \Delta ; \Gamma \rangle \vdash C \endprooftree
\end{center}

\item[$\circ$] Redex$_{\llto}$: introduction $\llto_i$ and direct elimination of $\llto_e$.
\begin{center}
\prooftree 
	\prooftree \leadsto  \Delta \vdash D \using \delta_1 \proofdotnumber = 5\endprooftree 
	\quad \prooftree 
		\prooftree \leadsto (  D , \Gamma ) \vdash C\endprooftree 
		\justifies
		\Gamma \vdash D \llto  C
		\using
	         	    [\llto_{i}]
         		\thickness=0.07em
	\endprooftree 
	\justifies
	(\Delta , \Gamma )\vdash C
	\using
             [\llto_{e}]
         \thickness=0.07em
\endprooftree 
$\quad \Rightarrow\quad$
\prooftree \prooftree   \leadsto \Delta \vdash D \using \delta_1 \proofdotnumber=5\endprooftree  \leadsto  (\Delta , \Gamma) \vdash C\endprooftree 
\end{center}

\item[$\circ$] Redex$_{\odot}$: introduction $\odot_i$ and direct elimination of $\odot_e$ on the left.
\begin{center}
\prooftree 
	\prooftree 
		\prooftree  \leadsto  \Delta_1 \vdash A \using \delta_1 \proofdotnumber=5 \endprooftree 
		\prooftree  \leadsto  \Delta_2 \vdash B \using \delta_2 \proofdotnumber=5 \endprooftree 
		\justifies
		\langle \Delta_1; \Delta_2\rangle \vdash A\odot B
		\using
	         	    [\odot_{i}]
         		\thickness=0.07em
	\endprooftree 
	\quad \prooftree \leadsto  \Gamma [\langle A;B \rangle] \vdash D\endprooftree 
	\justifies
	\Gamma [\langle \Delta_1; \Delta_2\rangle]\vdash D
	\using
             [\odot_{e}]
         \thickness=0.07em
\endprooftree 
$\quad \Rightarrow\quad$
 \prooftree 
 	\Gamma [ \langle \prooftree \leadsto  A \using \delta_1 \proofdotnumber=5\endprooftree ; \prooftree  \leadsto  B \using \delta_2 \proofdotnumber=5\endprooftree \rangle]  \vdash  D
	\thickness=0 em
\endprooftree 
\end{center}

\item[$\circ$] Redex$_{\odot}$: introduction $\odot_i$ and direct elimination of $\odot_e$ on the right.
\begin{center}
\prooftree 
	 \prooftree \leadsto \Gamma \vdash A \odot B \using \delta_1 \proofdotnumber=5\endprooftree 
	\quad \prooftree 
		A \vdash A
		\quad B \vdash B
		\justifies
		\langle A;B \rangle \vdash A\odot B
		\using
	         	    [\odot_{i}]
         		\thickness=0.07em
	\endprooftree 
	\justifies
	\Gamma \vdash A\odot B
	\using
             [\odot_{e}]
         \thickness=0.07em
\endprooftree 
$\quad \Rightarrow\quad$
 \prooftree \leadsto  \Gamma \vdash A\odot B \using \delta_1 \proofdotnumber=5 \endprooftree
\end{center}

\item[$\circ$] Redex$_{\otimes}$: introduction $\otimes_i$ and direct elimination of $\otimes_e$ on the left.
\begin{center}
\prooftree 
	\prooftree 
		\prooftree \leadsto  A \using \delta_1 \proofdotnumber=5 \endprooftree 
		\prooftree \leadsto  B \using \delta_2 \proofdotnumber=5\endprooftree 
		\justifies
		A\otimes B
		\using
	         	    [\otimes_{i}]
         		\thickness=0.07em
	\endprooftree 
	\quad \prooftree A \quad B \leadsto  D\endprooftree 
	\justifies
	D
	\using
             [\otimes_{e}]
         \thickness=0.07em
\endprooftree 
$\quad \Rightarrow\quad$
 \prooftree 
 	\Gamma [ (\prooftree \leadsto  A \using \delta_1 \proofdotnumber=5\endprooftree , \prooftree  \leadsto  B \using \delta_2 \proofdotnumber=5\endprooftree)]  \vdash  D
	\thickness=0 em
\endprooftree 
\end{center}
\item[$\circ$] Redex$_{\otimes}$: introduction $\otimes_i$ and direct elimination of $\otimes_e$ on the right.
\begin{center}
\prooftree  
	\prooftree  \leadsto  \Gamma \vdash A\otimes B  \using \delta_1 \proofdotnumber=5 \endprooftree
	\quad \prooftree 
		A \vdash A
		\quad B \vdash B
		\justifies
		(A,B) \vdash A \otimes B
		\using
	         	    [\otimes_{i}]
         		\thickness=0.07em
	\endprooftree 
	\justifies
	\Gamma \vdash A \otimes B
	\using
             [\otimes_{e}]
         \thickness=0.07em
\endprooftree 
$\quad \Rightarrow\quad$
\prooftree \leadsto  \Gamma \vdash A\otimes B \using \delta_1  \proofdotnumber=5\endprooftree
\end{center}
\end{enumerate}

Again, we use the notion of $k$-extended-redex, that we give back to this calculus.
Every path of a principal branch $B(S_0)$ of lenght $k$ from $S_0$ to $S_n$ with $|S_0|^r=|S_n|^r$, such that $|S_0|$ is an elimination rule $R_e$ and $S_n$ is the conclusion of an introduction rule $R_i$ is called a $\mathbf{k}$\textbf{-extended--redex}. 
Note that \emph{0-extended-redex} are the usual redexes of PCMLL, presented above.

A proof is in \textbf{normal form} if it does not contain any $k$-extended-redexes,  $\forall k \in \N$.

\subsection{Normalisation of PCMLL}

A proof is in  \textbf{normal form} is it does not contain any $k$-extended-redex.
In the same way for L$_{\odot}$, we define the different parts of a measure used for the normalisation. 

\begin{enumerate}
\item For a rule $R$, \emph{implicative elimination} ($\lto_e$, $\lfrom _e$ or $\llto_e$), with $S_0$ as conclusion, 
the integer $e(R)$ is $k$ if there is a $k$-extended-redex in $B(S_0)$ called a $k$-extended-redex over $R$, and otherwise $0$ .

In this measure, we count every rules, including entropy rules and product eliminations.
(this is exactly the number $k$ for $k$-extended-redexes wich contains implicative eliminations).

\item 
For a rule $R$, \emph{product elimination}, with $S_0$ as conclusoin, the integer $g(R)$ is $k$ if there is a $k$-extended-redex in $B(S_0)$ called a $k$-extended-redex over $R$ and otherwise $0$.
(this is exactly the number $k$ for $k$-extended-redexes which contains product elimination).
\end{enumerate}

Let $\delta$ a proof of PCMLL, 
we define $IER(\delta)$ (\resp $PER(\delta)$) as the number of occurrences of implicative elimination rules (\resp product) in $\delta$.

We define $e(\delta)$ as $\min_{R \in IER(\delta)}(e(R))$ and $g(\delta)$ as $\min_{R \in PER(\delta)}(g(R))$ egal to $0$ if and only if $\delta$ does not contain implicative $k$-extended-redex any more (\resp product).
Moreover, we define $r(\delta)$ as the number of rule in $\delta$.
We introduce the measure of the proof $\delta$, wrote $|\delta|$, as the triplet of integers, with the lexicographic order:
$$<r(\delta), e(\delta), g(\delta)>$$

\begin{propri}\label{prop00}
A proof with the measure $\langle n, 0, 0 \rangle$ is in normal form.
\end{propri}

\begin{proof}
Let $\delta$ a proof of PCMLL, 
in the measure of $\delta$:
\begin{itemize}
\item[$\circ$] the first integer is the number of rules which is the number of rules which is minimal if $\delta$ does not contain any $k$-extended-redex.

\item[$\circ$] the second integer is the distance between each part of an implicative $k$-extended-redex ($\lto$, $\lfrom $ or $\llto$). If its value is null, $\delta$ contains no more.

\item[$\circ$] the third integer id the distance between each part of an product $k$-extended-redex ($\odot$ ou $\otimes$). If its value is null, $\delta$ contains no more.
\end{itemize}

Note that the two other redexes could only be $0$-extended-redexes.
The figure \ref{exformnorm2} shows two examples of proofs which are not in normal form.
\end{proof}

\begin{figure}[htbp]
\begin{center}
$\begin{array}{l|l}
$example 1$& $example 2$\\
\prooftree
	\vdash C
	\prooftree
		\vdash E \odot F 
		\quad \prooftree
			(C, \langle E; F\rangle) \vdash A \lfrom  B
			\justifies
			\langle E; F\rangle \vdash C \lto (A \lfrom  B)
			\using
		             [\lto_{i}]
         			\thickness=0.07em
		\endprooftree
		\justifies
		\vdash C \lto (A \lfrom  B)
		\using
		   [\odot_{e}]
	\endprooftree
	\justifies
	\vdash  A \lfrom  B
	\using
	   [\lto_{e}]
\endprooftree
&
\prooftree
	\prooftree
		\vdash A\otimes B
		\quad \prooftree
			A \vdash E
			\quad B \vdash F
			\justifies
			(A,B) \vdash E \otimes F
			\using
		             [\otimes_{i}]
         			\thickness=0.07em
		\endprooftree
		\justifies
		\vdash E \otimes F
		\using
		   [\otimes_{e}]
	\endprooftree
	(E,F) \vdash D
	\justifies
	\vdash  D
	\using
	   [\otimes_{e}]
\endprooftree
\\
e(\lto_e) = 1 \Rightarrow& g(\odot_e) = 1\Rightarrow\\
\end{array}$
\end{center}
\caption{Examples of proofs of PCMLL which are not in normal form \label{exformnorm2}}
\end{figure}

\begin{propri}\label{propkredetmin}
A $k$-extended-redex $S_0 \cdots S_k$, with an implicative elimination rule contains a $k'$-extended-redex, with $k>k'$.
\end{propri}

\begin{proof}
Let $\delta$ a proof in normal form.
It goes through a principal branch from the conclusion and we exhibit the possibilities encountered in this derivation:

One of minimal $k$-extended-redexes has the following structure:

\begin{center}
\prooftree
	\prooftree
		\prooftree	
			\prooftree	
				\prooftree
					\prooftree	
						\prooftree	
							\justifies
							X
							\using
							  [introduction]
							\thickness=0.07em
						\endprooftree	 
						\leadsto
						U
						\using \delta_3
						\proofdotnumber=5
					\endprooftree			
					\justifies
					U \lfrom  A
					\using
					  [\lfrom _i]
					\thickness=0.07em
				\endprooftree
				\leadsto
				U \lfrom  A
				\using \delta_2
				\proofdotnumber=5
			\endprooftree
			\quad A
			\justifies
			U
			\using
			  [\lfrom _e]	
         			\thickness=0.07em
		\endprooftree
		\leadsto
		X
		\using \delta_1
		\proofdotnumber=5
		\thickness=0.07em
	\endprooftree
	\justifies
	\using
	  [elimination]
	\thickness=0.07em
\endprooftree
\end{center}
Then, we define: 
\begin{itemize}
\item $\delta_1$ as a sequence of implicative elimination rules and entropy;
\item $\delta_2$ as a sequence of product elimination rules and entropy.
\end{itemize}

$U$ is the formula produced by the higher implicative elimination rule. For this derivation, the number of symbols in $U$ is over the number os symbols in $X$
Then, in the principal branch, $\delta_2$ is a sequence of product eliminations and entropy.

The only rule above which could give the formula $U \lfrom  A$ is an introduction rule, because it is the only type which decrease the number of symbols in the formula.

On examples, the only introduction rule that we could structurally use is $\lfrom _i$ on the formula $A$.
Because this introduction is in the principal branch, it must be the conjoined rule of the previous introduction.
Then, we have a new $k'$-ectended-redex inside the $k$-extended-redex.
$k'$ is the number of rules in $\delta_2$ and because $\delta_2$ is a sub-part of the full proof, $k>k'$.
\end{proof}

We could propose a consequence of this property.

\begin{lem}
\label{propricycle}
In the proof $\delta$, a $k$-extended-redex which minimize $e(\delta)$ nonzero only contains product elimination rules and entropy.
\end{lem}

\begin{proof}
If the $k$-extended-redex is minimal, with the property \ref{propkredetmin}, there is no elimination rule. Moreover, if we do not use elimination rules, the number of symbols in the formula could not decrease. 
It must be constant in the $k$-extended-redex.

In this case, the only rules that we could use are rules with the conclusion is one of the premises.
The sequence of rules could only contains product elimination and entropy rules: $\odot_e, \otimes_e$ or $\sqsubset$.
\end{proof}

\begin{propri}\label{propriproductbelow}
Product eliminations and entropy rules could go under implicative elimination rules.

Let $R$ a product elimination $\otimes_e$ (\resp a rule $\odot_e$) of $\Gamma[\Delta]\seq C$ between a proof $\delta_0$ with $\Delta\seq A\otimes B$ as conclusion and a proof $\delta_1$ with $\Gamma[(A,B)]\seq C$ as conclusion (\resp $\Gamma[\langle A;B\rangle]\seq C$). This proof is combined with a proof $\delta_2$  with $\Theta \seq U$ as conclusion, by an implicative elimination rule $R'$. Then, the conclusion is $\langle\Theta ; \Gamma[\Delta]\rangle \seq V$ if $R'$ is $\lto_e$ (\resp $\langle\Gamma[\Delta]; \Theta\rangle \seq V$ if $R'$ is $\lfrom _e$ and $(\Gamma[\Delta] , \Theta ) \seq V$ if $R'$ is $\llto_e$).
Figure \ref{descprod} presents the case where $R'$ is $\lto_e$.

Then, we can obtain a proof for the same sequent which depends on $R'$ by applying first the rule $R'$ between the proof $\delta_2$ with $\Theta \seq U$ as conclusion and the proof $\delta_1$ with $\Gamma[(A,B)]\seq X$ as conclusion (\resp $\Gamma[\langle A;B\rangle]\seq X$) giving the sequent $\langle \Theta ; \Gamma[(A,B)]\rangle\seq V$ (\resp $\langle \Gamma[\langle A;B\rangle] ; \Theta\rangle\seq V$ and $( \Theta , \Gamma[(A,B)])\seq V$). Applying the rule $R$ on this new proof, we get the same sequent $\langle \Theta ; \Gamma[\Delta]\rangle \seq V$ (\resp $\langle \Gamma[\Delta] ; \Theta\rangle\seq V$ and $( \Theta , \Gamma[\Delta])\seq V$).
\end{propri}

\begin{figure}[htbp]
\begin{center}

\prooftree
	\prooftree
		\leadsto
		\Theta \seq U
		\using
		  \delta_2
	\endprooftree
	\quad 
	\prooftree
		\prooftree
			\leadsto
			\Delta \seq A\otimes B
			\using
			  \delta_0
		\endprooftree
		\prooftree
			\leadsto
			\Gamma[(A;B)] \seq C
			\using
			  \delta_1
		\endprooftree
		\justifies
		\Gamma[\Delta] \seq C
		\using R
	\endprooftree
	\justifies
	\langle \Theta ; \Gamma[\Delta]\rangle \seq V
	\using [\lto_e]
\endprooftree
$\quad \Rightarrow\quad $
\prooftree
	\prooftree
		\leadsto
		\Delta \seq A\otimes B
		\using
		  \delta_0
	\endprooftree
	\prooftree
		\prooftree
			\leadsto
			\Theta \seq U
			\using
			  \delta_2
		\endprooftree
		\quad
		\prooftree
			\leadsto
			\Gamma[(A;B)] \seq C
			\using
			  \delta_1
		\endprooftree
		\justifies
		\langle \Theta ; \Gamma[(A;B)]\rangle \seq V
		\using [\lto_e]
	\endprooftree
	\justifies
	\langle \Theta ; \Gamma[\Delta]\rangle \seq V
	\using R
\endprooftree

\end{center}
\caption{The go down of the product on the rule $\lto_e$ in PCMLL. \label{descprod}}
\end{figure}

\begin{proof}
Implicative eliminations do not modify the order between formulas of a same premise and do not use them.
Product elimination and entropy rules do not modify formulas but only hypothesis.
Then, these rules could be used in any order.
\end{proof}

\begin{theo}
Every proof $\delta$ of PCMLL has a normal form.
\end{theo}

\begin{proof}Let $\delta$ a proof such that $|\delta|= \langle n, e, g \rangle$.

We proceed by induction on the measure of proof.

By induction, every proof $\delta'$ of measure $|\delta '| < \langle n, e, g \rangle$ has a normal form.

\noindent If  $\delta$ contains a redex: reducing this redex reduces the number of rules in $\delta$, then the resulting proof $\delta'$ is such that $n(\delta') < n(\delta)$, hence $|\delta'| < |\delta|$. By induction $\delta'$ has a normal form, then $\delta$ also.\\

\noindent Else:
\begin{verse}
If $e(\delta) \neq 0$, then there is an implicative elimination rule $S$ such that $S$ is in a $e(\delta)$-extended-redex. 
This $e(\delta)$-extended-redex must be minimal, and the property \ref{propricycle} implies that it contains only product elimination and entropy rules.
Moreover, the property \ref{propriproductbelow} allows to rise $S$ over the rule above it (which corresponds to lower a product elimination or entropy rules below an implicative elimination).
The proof obtained $\delta'$ is such that $n(\delta') = n(\delta)$ and $e(\delta') = e(\delta) - 1$. 
The measure reduces and the induction allows to conclude that $\delta'$ has a normal form and then $\delta$.

Else $e(\delta) = 0$:

\begin{verse}
If $g(\delta)\neq0$: 
then it exists a product elimination rule $R$ such that $R$ is in a $g(\delta)$-extended-redex.
In this case, $\delta$ does not contain any implicative extended-redex and is it is minimal, it only contains product elimination and entropy rules.
$R$ could rise over the left premise -- the rule conjoined $\otimes_i$ is necessary in this part of the proof.

But for those rules, product eliminations can still rise above as shown in the property \ref{product-upwards}.

Then, we obtain the proof $\delta'$ such that $n(\delta') = n(\delta)$. 
We check that none new $k$-extended-redex based on a rule of $IEP(\delta)$ appears:

$\delta$ has the form:

\begin{center}
\prooftree
	\prooftree
		\prooftree
			\leadsto
			X
			\using
			  \delta_2
			\proofdotnumber=7
		  \endprooftree
		\quad 
		\prooftree
			A \odot B
			\leadsto
			A \odot B
			\using
			  \delta_3
			\proofdotnumber=7
		\endprooftree	
		\justifies
		A \odot B
		\using
		  [R]
		\thickness=0.07em	
	\endprooftree
	\quad 
	\prooftree
		\leadsto
		D
		\using
		  \delta_1
		\proofdotnumber=7
	\endprooftree
	\justifies
	\prooftree
		D
		\leadsto
		\using
		  \delta_4
		\proofdotnumber=7
	\endprooftree
	\using
	  [\odot_e]
	\thickness=0.07em	
\endprooftree
\end{center}

\begin{itemize}
\item[$\circ$] every principal branch of $\delta_3$ to $\delta_4$ does not contain extended-redexes because $\delta_3$ is in the left part of $\odot_e$
\item[$\circ$] every principal branch of  $\delta_1$ to $\delta_4$ can contain extended-redexes.
\item[$\circ$] every principal branch of  $\delta_2$ to $\delta_4$ does not contain extended-redexes because $\delta_2$ is in the left part of $\odot_e$ (only for product elimination rules).
\end{itemize}

The reduction scheme of the redex then gives the new structure of the proof $\delta'$:

\begin{center}
\prooftree
	\prooftree
		\leadsto
		X
		\using
		  \delta_2
		\proofdotnumber=7
	\endprooftree
	\prooftree
		\prooftree
			A \odot B
			\leadsto
			A \odot B
			\using
			  \delta_3
			\proofdotnumber=7
		\endprooftree
		\quad
		\prooftree
		\leadsto
		D
		\using
		  \delta_1
		\proofdotnumber=7
	\endprooftree
	\justifies
		D
		\using
		  [\odot_e]
		\thickness=0.07em	
	\endprooftree
	\justifies
		\prooftree
			D
			\leadsto
			\using
			  \delta_4
			\proofdotnumber=7
		\endprooftree
		\using
		  [R]
		\thickness=0.07em	
\endprooftree
\end{center}

In this new proof:
\begin{itemize}
\item[$\circ$] every principal branch of $\delta_3$ to $\delta_4$ does not contain extended-redexes because $\delta_3$ is in the flet par of $\odot_e$
\item[$\circ$] every principal branch of $\delta_1$ to $\delta_4$ does not contain new extended-redexes, and the measure for these extended-redexes decrease of 1.
\item[$\circ$] every principal branch of $\delta_2$ to $\delta_4$ does not contain extended-redexes because $\delta_2$ is in the left part of $R$ (which is necessary a product elimination rule).
\end{itemize}
The proof does not contain new $k$extended--redex ; then the proof does not contain new implicative extended-redex.
Then, we have $e(\delta') = e(\delta)$ and $g(\delta') = g(\delta) -1$.
Thus $|\delta'| < |\delta|$ and by induction $\delta'$ has a normal form, and therefor $\delta$ also has one.

Else: we have $e(\delta) = g(\delta) = 0$, because of the property \ref{prop00}, $\delta$ is in normal form.
\end{verse}
\end{verse}
\end{proof}

Normal forms of the previous examples are presented in the figure \ref{exformnorm2}. 
they are derived from the underlying algorithm of the previous proof.

\begin{center}
$\begin{array}{l|l}
$example 1$& $example 2$\\
\prooftree
	\vdash E \odot F 
	\prooftree
		\vdash C
		\quad \prooftree
			(C, \langle E; F\rangle) \vdash A \lfrom  B
			\justifies
			\langle E; F\rangle \vdash C \lto (A \lfrom  B)
			\using
		             [\lto_{i}]
         			\thickness=0.07em
		\endprooftree
		\justifies
		\vdash  A \lfrom  B
		\using
		   [\lto_{e}]
	\endprooftree
	\justifies
	\vdash C \lto (A \lfrom  B)
	\using
	   [\odot_{e}]
\endprooftree
&
\prooftree	
	\vdash A\otimes B
	\prooftree
		 \prooftree
			A \vdash E
			\quad B \vdash F
			\justifies
			(A,B) \vdash E \otimes F
			\using
		             [\otimes_{i}]
	         		\thickness=0.07em
		\endprooftree	
		(E,F) \vdash D
		\justifies
		(A,B) \vdash  D
		\using
		   [\otimes_{e}]
	\endprooftree
	\justifies
	\vdash D
	\using
	   [\otimes_{e}]
\endprooftree
\\

e(\lto_e) = 0& g(\odot_e) = 0\\
\end{array}$
\end{center}

Now, for proofs in normal form, we check the sub-formula property.

\subsection{Sub-formula property for PCMLL}

\begin{theo}
The sub-formula property holds for PCMLL: in a normal proof $\delta$ of a sequent $\Gamma \vdash C$, every formulae of a sequent is a sub-formula of some hypothesis ($\Gamma$) or of the conclusion ($C$). 
\end{theo}

\begin{proof} We proceed by induction on the number of rules of the proof.
Once again, we use a stronger definition of the property:

every formulae in a normal proof are sub-formulae of some hypotheses or of the conclusion of the proof and if the last rule used is an implicative elimination $\lto_e$, $\lfrom _e$ or $\llto_e$ every sub-formulae are sub-formulae of some hypotheses only.

Remark that entropy conserve the premise as conclusion, then entropy checked the sub-formula property for every case. And moreover, the axiom rule checks the property because the sequent is an hypothesis.

We check the induction hypothesis after the use of each other rule:

\begin{enumerate}
\item for the rule $\lto_e$.
\begin{center}
\prooftree 
	\prooftree \Delta_1  \leadsto \Delta \vdash C \using  \delta_1 \endprooftree 
	\quad \prooftree  \Gamma_2 \leadsto \Gamma \vdash C \lto D \using \delta_2\endprooftree 
	\justifies
	< \Delta ; \Gamma > \vdash D
	\using
	  [\lto_e]
         \thickness=0.07em
\endprooftree 
\end{center}

By hypothesis induction, every formulae in $\delta_1$ are sub-formulae of hypotheses $\Delta_1$ or of the conclusion $C$. And every formulae  in $\delta_2$are sub-formulae of hypotheses $\Gamma_2$ or of the conclusion $C \lto D$.
However $C$ is  sub-formula of $C \lto D$ and $D$ too.
One needs that $C \lto D$ was sub-formula of hypothesis of $\delta_2$ hypotheses.

Let us look the rule $R$ which produce $C \lto D$:
\begin{itemize}
\item if R is $\lto_i$: impossible because it is a $0$-extended-redex and the proof is in normal form.

\item if R is $\lfrom _i$, $\llto_i$, $\otimes_i$ or $\odot_i$: these cases are structurally impossibles beause because they can not produce $C \lto D$.

\item if R is $\lto_e$, $\lfrom _e$ or $\llto_e$: we use the induction hypothesis and $C\lto D$ is sub-formula of $\Gamma_2$.

\item if R is $\otimes_e$, $\odot_e$ or entropy: keep as conclusion the premise, thus we have to check the rule above.

If it is one of the previous rule, we use the same argument.
Else, the proof is a finite sequence of $\otimes_e$, $\odot_e$ and entropy.
They keep as conclusion a premise and then $C\lto D$ is one of the hypothesis of $\Gamma_2$.
\end{itemize}

\item it is strictly symmetrical for $\lfrom _e$ and $\llto_e$:
\bigskip

\begin{center}
\prooftree 
	\prooftree \Gamma_2 \leadsto \Gamma \vdash D \lfrom  C \using \delta_2\endprooftree 
	\quad \prooftree \Delta_1 \leadsto \Delta \vdash C \using \delta_1\endprooftree 
	\justifies
	<\Gamma ; \Delta > \vdash D
	\using
	  [\lfrom _e]
         \thickness=0.07em
\endprooftree 
$\quad$
\prooftree 
	\prooftree \Delta_1  \leadsto  \Delta \vdash C \using \delta_1\endprooftree 
	\quad \prooftree \Gamma_2 \leadsto  \Gamma \vdash C \llto D \using \delta_2\endprooftree 
	\justifies
	< \Delta ; \Gamma > \vdash D
	\using
	  [\llto_e]
         \thickness=0.07em
\endprooftree 
\end{center}
Every formulae are sub-formulae of $D\lfrom C$ (\resp $C\llto D$) or of hypotheses and $D\lfrom C$ (\resp $C\llto D$) is sub-formula of hypothesis $\Gamma_2$.

\item For all the introduction rules, the conclusion of each part of the proof are sub-formula of the conclusion. 
If the rule R is $\lto_i$,
Let the proof $\delta$:

\begin{center}
\prooftree 
	\prooftree \Gamma \leadsto  <\Gamma ; C > \vdash D\using \delta_1\endprooftree 
	\justifies
	\Gamma \vdash C \lto D
	\using
	  [\lto_i]
         \thickness=0.07em
\endprooftree 
\end{center}

Using the induction hypothesis, every formula in $\delta_1$ are sub-formulae of hypothesis $\Gamma$ or conclusion $D$. 
However $D$ is sub-formula of $C \lto D$ thus in $d$ every formulae of $\delta$ are sub-formulae of $\Gamma$ or of the conclusion  $C \lto D$.

\item for the rule $\lfrom _i$, let the proof $\delta$:
\begin{center}
\prooftree 
	\prooftree \Gamma \leadsto  <\Gamma ; C > \vdash D \using \delta_1\endprooftree 
	\justifies
	\Gamma \vdash D \lfrom  C
	\using
	  [\lfrom _i]
         \thickness=0.07em
\endprooftree 
\end{center}
Every formulae in $\delta_1$ are sub-formulae of hypotheses $\Gamma$ ou of the conclusion $D$.
$D$ is sub-formula of $D \lfrom  C$, then every formula of $\delta$ is sub-formula of hypotheses $\Gamma$ or of the conclusion $D \lfrom  C$. The property is checked in $\delta$.

\item for the rule $\llto_i$, let theproof $\delta$:
\begin{center}
\prooftree 
	\prooftree \Gamma \leadsto  <\Gamma ; C > \vdash D \using \delta_1\endprooftree 
	\justifies
	\Gamma \vdash C \llto D
	\using
	  [\llto_i]
         \thickness=0.07em
\endprooftree 
\end{center}
Every formulae in $\delta_1$ are sub-formulae of hypotheses $\Gamma$ ou of the conclusion $D$.
$D$ is sub-formula of $D \llto C$, then every formula of $\delta$ is sub-formula of hypotheses $\Gamma$ or of the conclusion $D \lfrom  C$. The property is checked in $\delta$.

\item for the rule $\otimes_i$, let the proof $\delta$:
\begin{center}
\prooftree 
	\prooftree \Delta_1 \leadsto  \Delta \vdash C\using \delta_1\endprooftree 
	\quad \prooftree \Gamma_2 \leadsto  \Gamma \vdash D\using \delta_2\endprooftree 
	\justifies
	( \Delta , \Gamma ) \vdash C \otimes D
	\using
	  [\otimes_i]
         \thickness=0.07em
\endprooftree 
\end{center}

\begin{itemize}
\item every formulae in $\delta_1$ are sub-formulae of hypotheses $\Delta_1$ or of the conclusion $C$.
\item every formulae in $\delta_2$ are sub-formulae of hypotheses $\Gamma_2$ or of the conclusion $D$.
\item however $C$ and $D$ are sub-formulae of $C \otimes D$, then in $\Gamma$, every formulae are sub-formulae of hypotheses $\Delta$ and $\Gamma$ or of the conclusion $C \otimes D$.
\end{itemize}

\item for the rule $\otimes_e$ :
\begin{center}
\prooftree 
	\prooftree \Delta_1 \leadsto  \Delta \vdash A \otimes B\using \delta_1\endprooftree 
	\quad \prooftree \Gamma_2 \leadsto  \Gamma \vdash D\using \delta_2\endprooftree 
	\justifies
	D
	\using
	  [\otimes_e]
	\thickness=0.07em	
\endprooftree 
\end{center}
\begin{itemize}
\item every formulae of $\delta_1$ are sub-formulae of hypotheses $\Delta_1$ or of the conclusion $A\otimes B$.
\item every formulae of  $\delta_2$ are sub-formulae of hypotheses $\Gamma_2$ or of the conclusion $D$.
\item moreover, $D$ is the conclusion of $\delta$. Thus, every formulae of $\delta_2$ are sub-formulae of hypotheses $\Gamma_2$ or of the conclusion of the proof $\delta$: $D$.
\end{itemize}

To check that the property holds for the other part of the proof, we must prove that $A\otimes B$ is sub-formula of hypotheses of $\delta_1$.
We look at the rule R above:

\begin{itemize}
\item[$\circ$] if $R$ is $\lto_e$, $\lfrom _e$ or  $\llto_e$, using the induction hypothesis $A \otimes B$ is sub-formula of hypotheses $\Delta_1$. 
\item[$\circ$] if $R$ is $\otimes_i$: this case is impossible because it might be a $0$-extended-redex and the proof is in normal form.
\item[$\circ$] if $R$ is $\lto_i$, $\lfrom _i$, $\llto_i$ ou  $\odot_i$: these case are structurally  impossible becaue these rules can not produce $A\otimes B$.
\item[$\circ$] if $\otimes_e$, $\odot_e$ or entropy: they keep one premise as conclusion, we analyze the rue above:
\begin{itemize}
\item[$\bullet$] either there is one of the previous rules, thus using the same arguments we conclude.
\item[$\bullet$] either, the proof  containing a finite number of rules, the sequence of rules is finite. Moreover,  it contains only $\otimes_e$, $\odot_e$ and entropy rules, thus the formula is a hypothesis in $\Delta_1$.
\end{itemize}
In every possible case, $A\otimes B$ or $A\odot B$ is sub-formula of hypotheses.
\end{itemize}
\end{enumerate}
In PCMLL, all proofs have a normal form and enjoy the sub-formula property.
\end{proof}

\section{Conclusion}

Motivated by concurrency and computational linguistics we have been defining PCMLL in natural deduction and proved normalisation. For Lambek calculus with product,  a subcalculus of PCMLL, we also characterized the unique normal proof. We only sketched the normal form(s) for the complete calculus and  writting this is our next job. 

Next we'll look forward a proof net syntax for PCMLL, which also allows to easily compute lambda terms and semantic reading. Despite  existence of proof nets for MLL and the Lambek calculus (of which PCMLL is the superimposition), and for intuitionistic NL of Abrusci and Ruet, because of the more flexible entropy rule that we are using, there is not yet  any proof net calculus for PCNLL. This work can be viewed as a first step in this direction. 

The uniqueness of the normal form for PCMLL can be done because what distinguishes two proofs in normal form of the same proof, is the relative position of $\otimes_e$ in the sequences in which they belong.
In reality, there exists a unique normal form but not canonical.
The algorithm of normalization that we propose retains the starting order between the different $\otimes_e$
This gives a unique form.
However, proofs with inversions of $\otimes_e$ in sequences of $\otimes_e$ provide equivalent proofs.

An extension of this work to a version as proof network of PCMLL (which also provide $\lambda$-terms and semantic representations) is a later stage will continue this work.

With respect to computational linguistic application, we look forward a simpler translation from 
PCMLL formulae to arrow types on $e$ and $t$ and thus from parse structures that are PCMLL deduction  to intuitionistic deduction, which are semantic readings. This open the question to interpret noun phrase and generalised quantifiers as the combination of $k$ (case) and $d$ (entities). The linear logic view on syntax also introduce the possibility to use game semantics, Girard's ludics, rather than model theoretic semantics for interpreting parse structures, and this ongoing work is part  a French national research program called Prelude.

\bibliographystyle{alpha}
\bibliography{normrev}
\end{document}